\documentclass[onecolumn,prx,amsmath,amssymb]{revtex4-2}

\usepackage[normalem]{ulem}

\usepackage{enumitem}
\usepackage{graphicx}
\usepackage{float}
\usepackage{amssymb,amsthm,amsfonts,amstext}
\usepackage{url}
\usepackage{color}
\usepackage{xcolor}
\usepackage{bbm}
\usepackage{hyperref}
\usepackage{changes}
\usepackage{lipsum}
\usepackage{mathrsfs}
\usepackage{qcircuit}

\usepackage{calrsfs}
\DeclareMathAlphabet{\pazocal}{OMS}{zplm}{m}{n}

\newtheorem{lem}{Lemma}

\newcommand{\R}{\mathbb{R}}

\def\be{\begin{equation}}
\def\ee{\end{equation}}
\def\bra#1{\langle#1|} \def\ket#1{|#1\rangle}
\def\braket#1#2{\langle#1|#2\rangle}

\def\ketbra#1#2{\ket{#1}\!\bra{#2}}
\def\proj#1{\ket{#1}\!\bra{#1}}

\def\id{{\mathbb I}}
\def\tr{\mbox{tr}}
\def\A{{\pazocal A}}
\def\C{{\pazocal C}}
\def\H{{\cal H}}

\def\norm#1{\| #1 \| }
\def\abs#1{|#1|}

\newtheorem{theo}{Theorem}

\newtheorem{prop}[theo]{Proposition}

\begin{document}

\title{Quantum theory based on real numbers can be experimentally falsified}
\author{Marc-Olivier Renou$^1$, David Trillo$^2$, Mirjam Weilenmann$^2$, Thinh P. Le$^2$, Armin Tavakoli$^{2,3}$, Nicolas Gisin$^{4,5}$, Antonio Ac\'in$^{1,6}$ and Miguel Navascu\'es$^2$}

\affiliation{$^1$ICFO-Institut de Ciencies Fotoniques, The Barcelona Institute of Science and Technology, 08860 Castelldefels (Barcelona), Spain\\
$^2$Institute for Quantum Optics and Quantum Information (IQOQI) Vienna, Austrian Academy of Sciences\\
$^3$Institute for Atomic and Subatomic Physics, Vienna University of Technology, 1020 Vienna, Austria\\
$^4$Group of Applied Physics, University of Geneva, 1211 Geneva, Switzerland\\
$^5$Schaffhausen Institute of Technology – SIT, Geneva, Switzerland\\
$^6$ICREA-Instituci\'o Catalana de Recerca i Estudis Avançats, Lluis Companys 23, 08010 Barcelona, Spain}

\begin{abstract}
While complex numbers are essential in mathematics, they are not needed to describe physical experiments, expressed in terms of probabilities, hence real numbers. Physics however aims to explain, rather than describe, experiments through theories. 
While most theories of physics are based on real numbers, quantum theory was the first to be formulated in terms of operators acting on complex Hilbert spaces~\cite{dirac,von_neumann}. This has puzzled countless physicists, including the fathers of the theory, for whom a real version of quantum theory, in terms of real operators, seemed much more natural \cite{einstein2011letters}. In fact, previous works showed that such `real quantum theory' can reproduce the outcomes of any multipartite experiment, as long as the parts share arbitrary real quantum states \cite{nicolas_real}.
Thus, are complex numbers really needed in the quantum formalism? Here, we show this to be case by proving that real and complex quantum theory make different predictions in network scenarios comprising independent states and measurements. This allows us to devise a Bell-like experiment whose successful realization would disprove real quantum theory, in the same way as standard Bell experiments disproved local physics. 
\end{abstract}
\maketitle
\begin{quote}
\emph{What is unpleasant here, and indeed directly to be objected to, is the use of complex numbers. $\Psi$ is surely fundamentally a real function.}

Letter from Schr\"{o}dinger to Lorentz \cite{einstein2011letters}. June $6^{th}$, 1926.    
\end{quote}

%Without further qualifications, this question must be answered in the negative: physics experiments are described by the statistics they generate, that is, by probabilities, and hence real numbers. There is therefore no need for complex numbers. The question however becomes meaningful when considering a specific theoretical framework, designed to explain existing experiments and make predictions about future ones. Whether complex numbers are needed within a theory to correctly explain experiments is not straightforward.

Are complex numbers necessary for natural sciences, and, more concretely, for physics? Without further qualifications, this question must be answered in the negative: physics experiments are described by the statistics they generate, that is, by probabilities, and hence real numbers. There is therefore no need for complex numbers. The question however becomes meaningful when considering a specific theoretical framework, designed to explain existing experiments and make predictions about future ones. Whether complex numbers are needed within a theory to correctly explain experiments, or whether real numbers only are sufficient, is not straightforward. Complex numbers are sometimes introduced in electromagnetism to simplify calculations: one might, for instance, regard the electric and magnetic fields as complex vector fields in order to describe electromagnetic waves. However, this is just a computational trick. Can we claim the same for quantum theory?

%Are complex numbers ``necessary" for physics? In electromagnetism, one often regard the electric and magnetic fields as components of a single complex vector field to compactify the formulation and simplify calculations. All physicists would agree that this is an aesthetic choice: the use of complex numbers is not required! Can one make the same conclusion for quantum theory and dispense completely the use of complex numbers in its formulation? As it turns out, this question is further complicated by the fact that quantum theory has many ``equivalent" formulations---Dirac-von Neumann, path-integral, and Bohmian---that looks rather different.

In its Hilbert space formulation, quantum theory is defined in terms of the following postulates~\cite{Encyclo, Piron}:
\begin{enumerate}[nosep]
\item[$(i)$] To every physical system $S$, there corresponds a Hilbert space $\H_S$ and its state is represented by a normalized vector $\phi$ in $\H_S$, that is, $\bra\phi\phi\rangle=1$.
\item[$(ii)$] A measurement $\Pi$ in $S$ corresponds to an ensemble $\{\Pi_r\}_r$ of projection operators acting on $\H_S$, with $\sum_r\Pi_r=\id_S$.
\item[$(iii)$] Born rule: if we measure $\Pi$ when system $S$ is in state $\phi$, the probability of obtaining result $r$ is given by $\text{Pr}(r)=\bra{\phi}\Pi_r\ket{\phi}$.
\item[$(iv$)] the Hilbert space $\H_{ST}$ corresponding to the composition of two systems $S,T$ is $\H_S\otimes\H_T$. The operators used to describe measurements or transformations in system $S$ act trivially on $\H_T$ and vice versa. Similarly, the state representing two independent preparations of the two systems is the tensor product of the two preparations. 
\end{enumerate}
This last postulate plays a key role in our discussions: we remark that it even holds beyond quantum theory, specifically for space-like separated systems in some axiomatizations of quantum field theory \cite{Werner,nuclearity, models_AQFT, models_AQFT2}, see the Appendix.

%Here our usage of ``quantum theory" will mean {\bf Hilbert space} quantum theory, which satisfies the following postulates:
%\begin{enumerate}[nosep]
%\item[$(i)$] To every physical system $S$, there corresponds a Hilbert space $\H_S$ and its state is represented by a normalized vector $\phi$ in $\H_S$, that is, $\bra\phi\phi\rangle=1$.
%\item[$(ii)$] A measurement $\Pi$ in $S$ corresponds to an ensemble $\{\Pi_r\}_r$ of projection operators acting on $\H_S$, with $\sum_r\Pi_r=\id_S$.
%\item[$(iii)$] Born rule: if we measure $\Pi$ when system $S$ is in state $\phi$, the probability of obtaining result $r$ is given by $\text{Pr}(r)=\bra{\phi}\Pi_r\ket{\phi}$.
%\item[$(iv$)] the Hilbert space $\H_{ST}$ corresponding to the composition of two independent systems $S,T$ is $\H_S\otimes\H_T$. The operators used to describe measurements or transformations in system $S$ act trivially on $\H_T$ and vice versa. Similarly, the state representing two independent preparations of the two systems is the tensor product of the two preparations. 
%\end{enumerate}
%We make no restriction on the choice of the number field: it could be real, complex, quaternion, etc. Also, the original Dirac-von Neumann formulation~\cite{dirac, von_neumann} explicitly states that the Hilbert space $\H_S$ in $(i)$ must be complex. We denote this version as $(i_{\mathbbm{C}})$, and call the theory $(i_{\mathbbm{C}}-iv)$ is the {\bf standard} quantum theory.%formulation of Hilbert space quantum theory, or shortly standard quantum theory.

As originally introduced by Dirac and von Neumann \cite{dirac, von_neumann}, the Hilbert spaces $\H_S$ in $(i)$ are traditionally taken to be complex. We call the resulting postulate $(i_{\mathbbm{C}})$. The theory specified by $(i_{\mathbbm{C}})$ and $(ii-iv)$ is the standard formulation of quantum theory in terms of complex Hilbert spaces and tensor products. For brevity, we will refer to it simply as `complex quantum theory'. Contrary to classical physics, complex numbers (in particular, complex Hilbert spaces) are thus an essential element of the very definition of complex quantum theory. 
%\deleted{We note that there are alternative formulations that recover the predictions of complex quantum theory, e.g.\ in terms of path integrals~\cite{Shankar_PrinciplesOfQM}, ordinary probabilities~\cite{Fuchs} and Wigner functions~\cite{Wigner}, which are beyond the scope of this work.} 

 Due to the controversy surrounding their irruption in mathematics and their almost total absence in classical physics, the occurrence of complex numbers in quantum theory worried some of its founders, for whom a formulation in terms of real operators seemed much more natural~\cite{einstein2011letters}. This is precisely the question we address in this work: can we replace complex by real numbers in the Hilbert space formulation of quantum theory without limiting its predictions? The resulting `real quantum theory', which has appeared in the literature under various names \cite{real_quantum1, real_quantum2}, obeys the same postulates $(ii-iv)$ but assumes real Hilbert spaces $\H_S$ in $(i)$, a modified postulate that we denote by $(i_{\mathbbm{R}})$. 
 
If real quantum theory led to the same predictions as complex quantum theory, then complex numbers would just be, as in classical physics, a convenient tool to simplify computations but not an essential part of the theory. However, we show that this is not the case: the measurement statistics generated in certain finite-dimensional quantum experiments involving causally independent measurements and  state preparations do not admit a real quantum representation, even if we allow the corresponding real Hilbert spaces to be infinite dimensional.

%\added{
Our main result applies to the standard Hilbert space formulation of quantum theory, through axioms ($i$)-($iv$). Note, though, that there are alternative formulations able to recover the predictions of complex quantum theory, e.g.\ in terms of path integrals~\cite{Shankar_PrinciplesOfQM}, ordinary probabilities~\cite{Fuchs}, Wigner functions~\cite{Wigner}, or Bohmian mechanics~\cite{Bohm}. 
For some formulations, e.g. \cite{stueckelberg, Aleksandrova_2013}, real vectors and real operators play the role of physical states and physical measurements respectively, but the Hilbert space of composed system is not a tensor product. While we briefly discuss some of these formulations in the Appendix, we do not consider them here because they all violate at least one of the postulates $(i_{\mathbbm{R}})$ and $(ii-iv)$. Our results imply that this violation is in fact necessary for any such model.
%}
%Beyond these postulates, note that alternative real descriptions of quantum theory predictions can be obtained, e.g. in terms of ordinary probabilities~\cite{Fuchs} and Wigner functions~\cite{Wigner}.

It is instructive to address our main question as a game between two players, the `real' quantum physicist Regina and the `complex' quantum physicist Conan. Regina is convinced that our world is governed by real quantum theory, while Conan believes that only complex quantum theory can describe it. Through a well chosen quantum experiment, Conan aims to prove Regina wrong; that is, to falsify real quantum theory by exhibiting an experiment that this theory cannot explain.

At first, Conan thinks of conducting simple experiments involving a single quantum system. Unfortunately, for any such quantum experiment, Regina can find a real quantum explanation. For instance, if $\rho$ is the complex density matrix that Conan uses to model his experiment, Regina could propose the state 
\begin{equation}
\label{eq:realstate}
\tilde{\rho} = \text{Re}(\rho)\otimes\frac{\id}{2}+\text{Im}(\rho)\otimes\frac{1}{2}\begin{pmatrix}0 &1\\-1 &0\end{pmatrix}=\frac{1}{2}(\rho\otimes\proj{+i} + \rho^*\otimes\proj{-i}),
\end{equation}
\noindent where $\ket{\pm i}=\frac{1}{\sqrt{2}}(\ket{0}\pm i\ket{1})$, and $*$ denotes complex conjugation. The operator $\tilde{\rho}$ is real and positive semidefinite: it is thus a real quantum state. Figure \ref{fig:scenarios} (left pane) explains how Regina can analogously define real measurement operators that, acting on $\tilde{\rho}$, reproduce the statistics of any (complex) measurement conducted by Conan on $\rho$. This construction is just one of the infinitely many ways that Regina has to explain the measurement statistics of any single-particle experiment using real operators, but it already implies that real quantum theory cannot be falsified in this scenario. It does not imply, however, that states in real quantum theory are restricted to have this form: they remain arbitrary, as in complex quantum theory.

Note that, assuming a fixed Hilbert space dimension, Conan could come up with single-site experiments where real and complex quantum theory differ, for instance because the former does not satisfy local tomography, or even leads to different experimental predictions, see e.g.~\cite{realresource}. However, since dimension cannot be upper bounded experimentally~\cite{dimwitn}, Regina would be right not to interpret any such experiment as a disproof of real quantum theory. In practice, any experimental system has infinite degrees of freedom: a finite dimension may just be an approximation made to simplify its theoretical description. 
Hence, to defeat Regina, Conan has to design an experiment in which no explanation using real Hilbert spaces is valid, no matter their dimension.

\begin{figure}
  \centering
  \includegraphics[width=12cm]{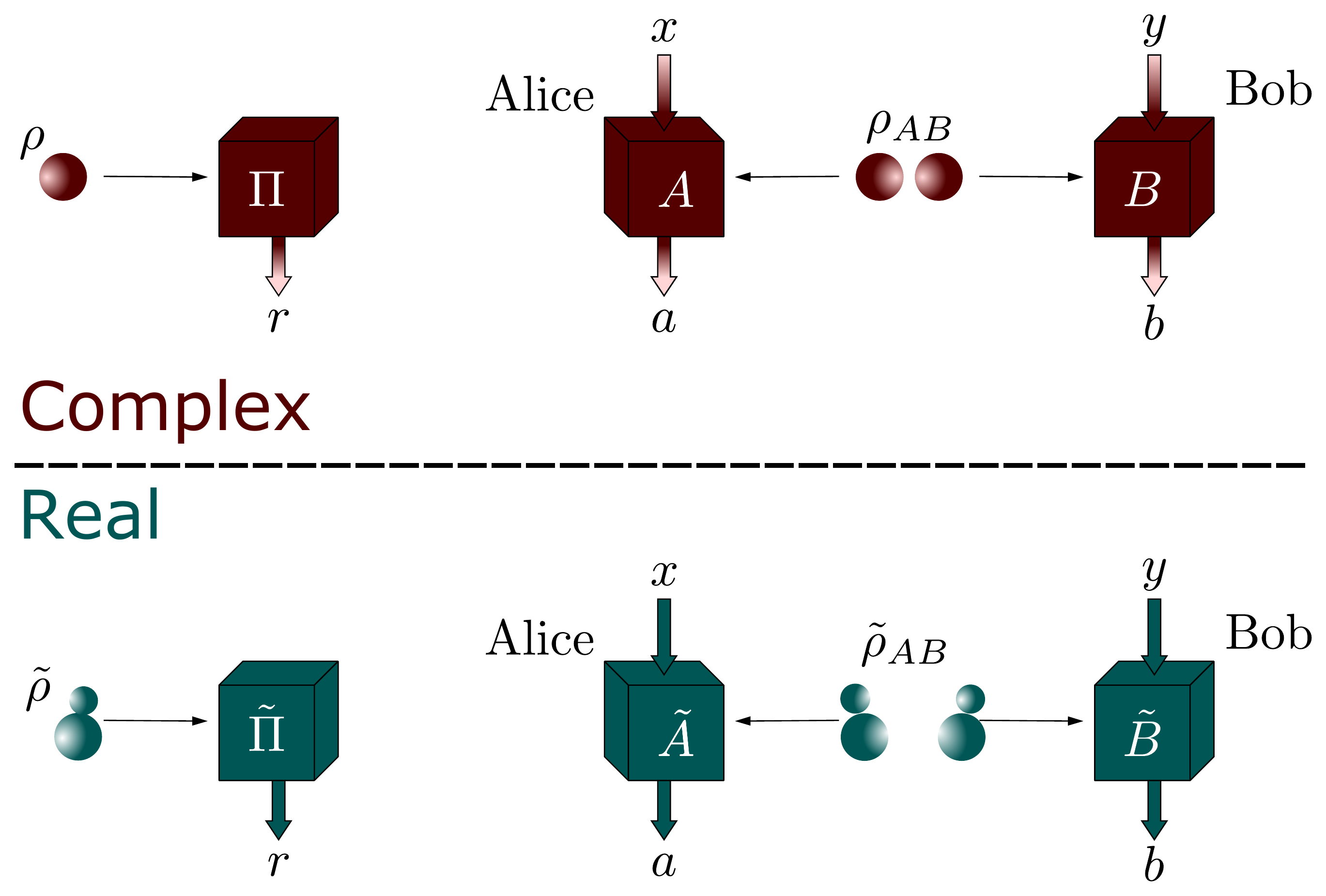}
  \caption{\textbf{(Left pane) Single complex quantum system.} Consider a single-site quantum experiment, where a system in state $\rho$ is probed the via measurement $\{\Pi_r\}_r$. One way to reproduce the measurement statistics of this experiment using  real quantum theory requires adding an extra real qubit: The state $\rho$ is then replaced by the real state $\tilde{\rho}$ in Eq.~\eqref{eq:realstate}, while every measurement operator is replaced by the real measurement operator $\tilde\Pi_r=\Pi_r\otimes\proj{i}+\Pi_r^*\otimes\proj{-i}$. Using the fact that probabilities are real, and thus $P(r)=P(r)^*=\tr(\rho^* \Pi^*_r)$, it is straightforward that $P(r)=\tr(\rho\Pi_r)=\tr(\tilde\rho\tilde\Pi_r)$. Note that this construction doubles the Hilbert space dimension of the original complex quantum system (when finite). This is not surprising, as a complex number is defined by two real numbers, and may just be seen as yet another example of how complex numbers simplify the calculation of experimental predictions, as in electromagnetism.
\textbf{(Right pane) Multipartite complex quantum system.} A complex Bell scenario consists of two particles (or systems) distributed between Alice and Bob, who perform local measurements, labelled by $x$ and $y$, and get results $a$ and $b$. By postulates $(i_{\mathbb{C}} - iv)$, a complex Hilbert space is assigned to each particle, and the Hilbert space describing the overall bipartite system is defined by the tensor product of these. The state of the two particles is thus described by an operator $\rho_{AB}$ acting on the joint space, while operators $A_{a|x}$ and $B_{b|y}$ acting on each local Hilbert space describe the local measurements. The observed measurement statistics or correlations are described by the conditional probability distribution $P(ab|xy)=\tr(\rho_{AB}A_{a|x}\otimes B_{b|y})$. One way to reproduce these statistics using {\em real quantum theory} consists in assigning an extra real qubit to each particle. The quantum state is replaced by the real state $\tilde \rho_{AA'BB'} = \frac{1}{2}(\rho_{AB}\otimes\proj{+i,+i}_{A'B'} + \rho_{AB}^*\otimes\proj{-i,-i}_{A'B'})$, and the local measurements are replaced by the same transformation as before for a single system. The observed statistics are again recovered, i.e., $P(ab|xy)=\tr(\tilde\rho_{AB}\tilde{A}_{a|x}\otimes\tilde{B}_{b|y})$.}
  \label{fig:scenarios}
\end{figure}

Conan may next consider experiments involving several distant labs, where phenomena like entanglement \cite{schroedinger} and Bell non-locality \cite{Bell1964} can manifest. 
For simplicity, we focus on the case of two separate labs. 
A source emits two particles (e.g.~a crystal pumped by a laser emitting two photons) in a state $\rho_{AB}$, each being measured by different observers, called Alice and Bob, see Fig.~\ref{fig:scenarios} (right pane). As pointed out by Bell \cite{Bell1964}, there exist quantum experiments where the observed correlations, encapsulated by the measured probabilities $P(a,b|x,y)$, are such that they cannot be reproduced by any local deterministic model. An experimental realization of such correlations disproves the universal validity of local classical physics.

Could Conan similarly refute real quantum theory via a (complex) quantum Bell experiment? Such an experiment should necessarily violate some Bell inequality; otherwise, one could reproduce the measured probabilities with diagonal (and hence real) density matrices and measurement operators. The mere observation of a Bell violation is, however, insufficient to disprove real quantum theory, as already exemplified by the famous Clauser-Horne-Shimony-Holt (CHSH) Bell inequality~\cite{chsh} $\text{CHSH}(x_1,x_2;y_1,y_2):=\langle A_{x_1}B_{y_1}\rangle +\langle A_{x_1}B_{y_2}\rangle+\langle A_{x_2}B_{y_1}\rangle-\langle A_{x_2}B_{x_2}\rangle\leq 2$. 
The inequality is derived for a Bell experiment where Alice and Bob perform two measurements with outcomes $\pm 1$, and where $A_x,B_y$ denote the results by Alice (Bob) when performing measurement $x,y$. The maximal quantum violation of this inequality is $\beta_\text{CHSH}=2\sqrt 2$ and Alice and Bob can attain it using real measurements on a real two-qubit state.

To find a gap between the predictions of real and complex quantum theory, Conan shall explore more complicated Bell inequalities. \emph{A priori}, promising candidates are the elegant inequality of Ref.~\cite{elegant} or the combination of three CHSH inequalities introduced in~\cite{optimal_ent_bit,self_testing}
\begin{equation}
\label{eq:chsh3}
%\text{CHSH}_3=\text{CHSH}(1,2;1,2)+\text{CHSH}(1,3;4,3)+\text{CHSH}(2,3;6,5)\leq 6,
\text{CHSH}_3 :=\text{CHSH}(1,2;1,2)+\text{CHSH}(1,3;3,4)+\text{CHSH}(2,3;5,6)\leq 6,
\end{equation}
designed for a scenario in which Alice and Bob perform three and six measurements respectively. The maximal violation of inequality (\ref{eq:chsh3}) is $3\beta_\text{CHSH}=6\sqrt 2$, which can be attained with complex measurements on qubits~\cite{self_testing}.

However, none of these Bell inequalities will work: as shown in~\cite{tamas_real, nicolas_real, moroder}, real quantum Bell experiments can reproduce the statistics of any quantum Bell experiment, even if conducted by more than two separate parties. Indeed, the construction of Eq.~\eqref{eq:realstate} for single complex quantum systems can be adapted to the multipartite case if we allow the source to distribute an extra qubit for each observer, see Figure \ref{fig:scenarios} (right pane) for details.

To defeat Regina, Conan may also look for inspiration to other no-go theorems in quantum theory, such as the Pusey-Barrett-Rudolph construction~\cite{PBR} involving states prepared in independent labs subject to joint measurements. Unfortunately, Regina is again able to provide an explanation to such scenarios using real quantum theory, see the Appendix. At this point, Conan might give up and accept that he will never change Regina's mind. He would not be alone. For years it was generally accepted that real quantum theory was experimentally indistinguishable from complex quantum theory. In other words: in quantum theory complex numbers are only convenient, but not necessary to make sense of quantum experiments. Next we prove this conclusion wrong.

All it takes for Conan to win the discussion is to go beyond the previous constructions and consider experimental scenarios where independent sources prepare entangled states to several parties, who in turn conduct independent measurements~\cite{bilocality,tobias1,tobias2,renouBSM,bancalBSM}. Such general network scenarios correspond to the future quantum internet, which will connect many quantum computers and guarantee quantum confidentiality over continental distances. Our results demonstrate how these networks, beyond their practical relevance, open radically new perspectives to solve open questions in the foundations of quantum theory when exploiting the causal constraints associated to their geometries.

To disprove real quantum theory, Conan proposes the network corresponding to a standard entanglement-swapping scenario, depicted in Fig.~\ref{fig:entswap}, consisting of two independent sources and three observers: Alice, Bob and Charlie. The two sources prepare two maximally entangled states of two qubits, the first one $\bar{\sigma}_{AB_1}$ distributed to Alice and Bob; and the second $\bar{\sigma}_{B_2C}$, to Bob and Charlie. Bob performs a standard Bell-state measurement on the two particles that he receives from the two sources. This measurement has the effect of swapping the entanglement from Alice and Bob and Bob and Charlie to Alice and Charlie: namely, for each of Bob's four possible outcomes, Alice and Charlie share a two-qubit entangled state. Note that the actual state depends on Bob's outcome, but not its degree of entanglement, which is always maximal. Alice and Charlie implement the measurements leading to the maximal violation of the $\text{CHSH}_3$ inequality~\eqref{eq:chsh3}. For these measurements, the state shared by Alice and Charlie, conditioned on Bob's result, maximally violates the inequality or a variant thereof produced by simple relabelings of the measurement outcomes.  

Regina takes up Conan's challenge and seeks to reproduce the statistics predicted by Conan. Since she works under the postulates $(i_{\mathbb{R}})$ and $(ii-iv)$, she models the experiment of Figure~\ref{fig:entswap} as follows: each subsystem is represented by a real Hilbert space $\H_S$ for $S=A,B_1,B_2,C$, the states of the two sources are arbitrary real density matrices acting on $\H_A\otimes\H_{B_1}$ and $\H_{B_2}\otimes\H_C$ respectively, and the arbitrary real measurements act on $\H_A,\H_{B_1}\otimes\H_{B_2},$ and $\H_C$ respectively. For each choice of states and measurements, she computes the probabilities via the Born rule. Regina's goal is to search over all states and measurements of the aforementioned form, acting on real Hilbert spaces of arbitrary dimension, until she can match Conan's predictions.

\begin{figure}
  \centering
  \includegraphics[width=15cm]{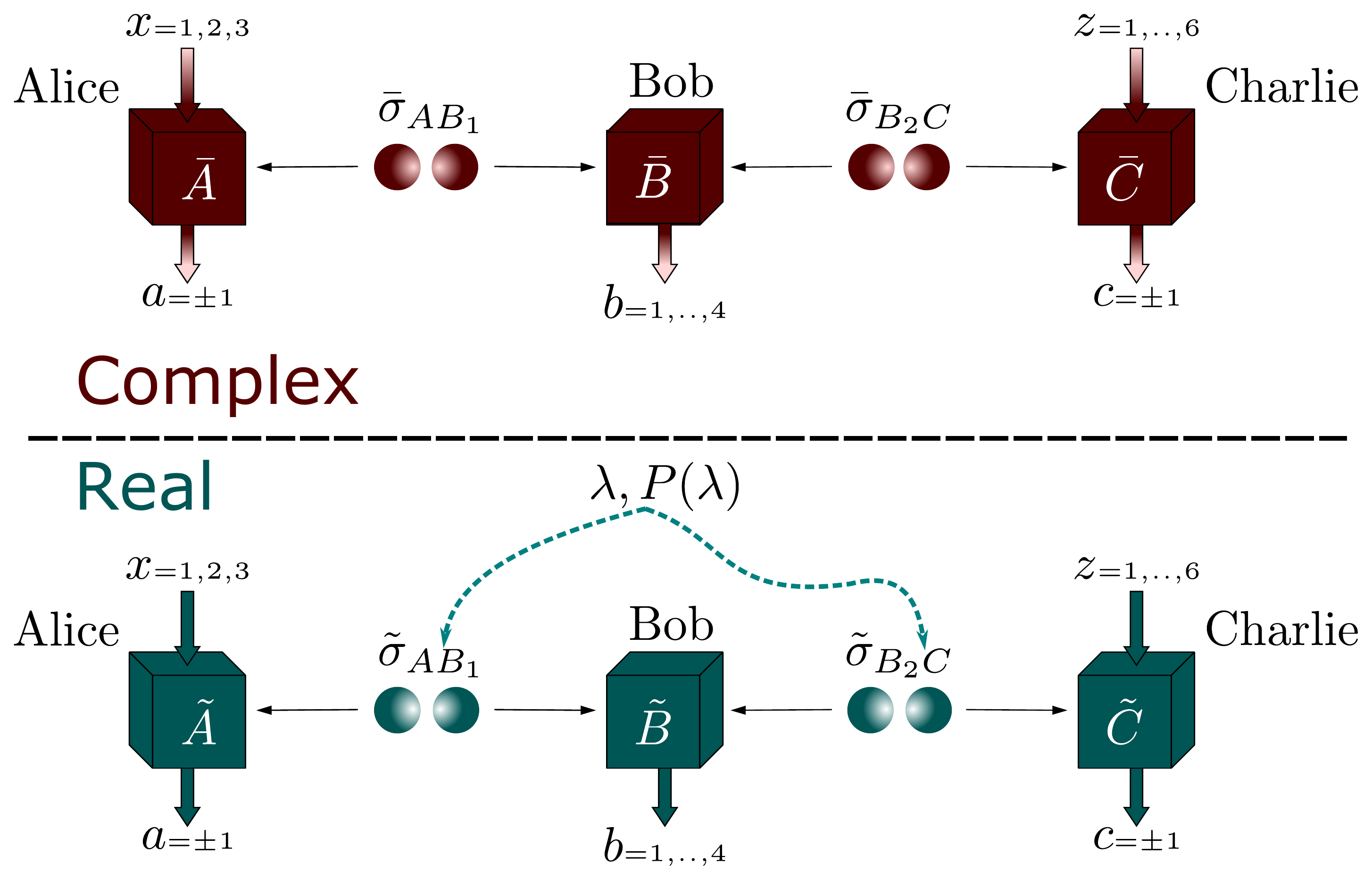}
  \caption{\textbf{Network scenario separating real and complex quantum theory.} In complex quantum theory (upper pane), two independent sources distribute the two-qubit states $\bar\sigma_{AB_1}$ and $\bar\sigma_{B_2C}$ to, respectively, Alice and Bob, and Bob and Charlie. At Bob's location, a Bell measurement, of four outputs, is implemented. Alice and Charlie apply the complex measurements leading to the maximal violation of the $\text{CHSH}_3$ inequality: three and six measurements with two possible outputs, labelled by $\pm 1$. According to quantum physics, the observed correlations read $\bar{P}(abc|xz)=\tr\left((\bar{\sigma}_{AB_1}\otimes\bar{\sigma}_{B_2C})(\bar{A}_{a|x}\otimes\bar{B}_{b}\otimes\bar{C}_{c|z})\right)$. These correlations cannot be reproduced, or even well approximated, when all the states and measurements in the network are constrained to be real operators of arbitrary dimension (lower pane). The impossibility still holds if the two preparations are correlated through shared randomness (dashed mosque arrows), resulting in correlations of the form $P(abc|xz)=\sum_\lambda P(\lambda)\tr\left((\tilde\sigma_{AB_1}^\lambda\otimes\tilde\sigma_{B_2C}^\lambda)(\tilde{A}_{a|x}\otimes\tilde{B}_{b}\otimes\tilde{C}_{c|z})\right)$, where all operators are real.}
  \label{fig:entswap}
\end{figure}

However, no construction by Regina is able to reproduce the measurement probabilities $\bar{P}(a,b,c|x,z)$ observed in the experiment. The proof, given in the Appendix, exploits the results of~\cite{self_testing}, where all quantum realizations leading to the maximal quantum value of~\eqref{eq:chsh3} were characterized. From this characterization, we show that the marginal state shared by Alice and Charlie at the beginning of the experiment cannot be decomposed as a convex combination of real product states \cite{real_entanglement}, as the network depicted in Fig.~\ref{fig:entswap} requires, and thus the statement follows. We moreover show the result to be robust, in the sense that the impossibility of real simulation also holds for non-maximal violations of the inequality~\eqref{eq:chsh3} between Alice and Charlie. This result settles the argument between Conan and Regina: since she cannot accommodate Conan's experimental observations within the real quantum framework, she must admit her defeat.

A different question now is whether it is experimentally feasible to disprove real quantum theory. To assess this, it is convenient to rephrase our impossibility result in terms of a Bell-type parameter, i.e., a linear function of the observed correlations. To this aim, we propose the Bell-type functional $\mathscr{T}$, defined by the sum of the violations of (the variants of) the $\text{CHSH}_3$ inequality for each of Bob's measurement outputs, weighted by the probability of the output. In the ideal entanglement-swapping realization with two-qubit maximally entangled states, the maximal quantum value of $\text{CHSH}_3$, equal to $6\sqrt 2$, is obtained for each of the four outputs by Bob, so $\mathscr{T}$ also attains its maximum quantum value, $\mathscr{T}=6\sqrt 2\approx 8.49$. In the Supplemental Material, we explain how to reduce the problem of upper bounding $\mathscr{T}$ to a convex optimisation problem, making use of the hierarchies~\cite{npa, npa2,NOP,moroder}, that we solve numerically \cite{yalmip,mosek}, for real quantum systems, to give
%Using a variant of the Navascu\'es-Pironio-Ac\'in (NPA) hierarchy~\cite{npa, npa2,NOP} introduced in~\cite{moroder}, we reduce the problem of upper bounding $\mathscr{T}$ over all real quantum realizations to a finite-dimensional convex optimization problem \cite{sdp}, that we solve numerically \cite{yalmip, mosek}, see the Appendix. 
 $\mathscr{T}\lesssim 7.66$. It remains open, whether this upper bound is tight.
%or the actual maximum value is even smaller. 
Since the map $\mathscr{T}$ is a linear function of the observed probabilities, the impossibility result holds even when the real simulation is assisted by shared randomness, see Fig.~\ref{fig:entswap} (lower pane). As shown in \cite{Elkouss}, \cite{Mateus}, this feature allows one to drop the assumption of independent and identical realizations in multiple-round hypothesis tests.

The setup needed to experimentally falsify real quantum theory is very similar to the bilocality scenario described in~\cite{bilocality}, for which several experimental implementations have been reported~\cite{sciarrino,pryde,pan, Baumer2020}. Beating the real bound on $\mathscr{T}$ requires the two distributed states to have each a visibility beyond $\sqrt{7.66/6\sqrt 2}\sim 0.95$, a value attained in several experimental labs worldwide. The experiment similarly relies on the implementation of a challenging \cite{Bell_meas_impossible} but feasible \cite{Bell_meas_possible} two-qubit entangled measurement. All things considered, we believe that an experimental disproof of real quantum physics based on the inequality $\mathscr{T}$ is within reach of current quantum technology, see the Appendix for more details.

Ever since the birth of modern science four centuries ago, abstract mathematical entities have played a big role in formalizing physical concepts. Our current understanding of velocity was only possible through the introduction of derivatives. The modern conception of gravity owes to the invention of non-Euclidean geometry. Basic notions from representation theory made it possible to formalize the notion of a fundamental particle. Here we considered whether the same holds for the complex numbers. Somewhat surprisingly, we found that there do exist natural scenarios that require the use of complex numbers to account for experimental observations within the standard Hilbert space formulation of quantum theory. As it turns out, some such experiments are within reach of current experimental capabilities, so it is not unreasonable to expect a convincing experimental disproof of real quantum theory in the near future.

From a broader point of view, our results advance the research program, started in~\cite{optimal_quantum}, of singling out quantum correlations by demanding maximal performance in a device-independent information-theoretic task. In this regard, our work shows that complex quantum theory outperforms real quantum theory when the non-local game $\mathscr{T}$ is played in the entanglement swapping scenario. This game can be interpreted as an extension of the adaptive CHSH game proposed in \cite{optimal_quantum}, which was recently shown to rule out a number of alternative physical theories in favor of quantum theory~\cite{optimal_quantum2b}. Whether the average score of $\mathscr{T}$ is maximized by complex quantum theory, or whether any physical theory other than complex quantum theory must necessarily produce a lower score are intriguing questions that we leave open.

\section*{Acknowledgements}
D.T. is a recipient of a DOC Fellowship of the Austrian Academy of Sciences at the Institute of Quantum Optics and Quantum Information (IQOQI), Vienna. M.W. and L.P.T. are supported by the Lise Meitner Fellowship of the Austrian Academy of Sciences (project numbers M 3109-N and M 2812-N respectively). M.-O.R. and A.T. are supported by the Swiss National Fund Early Mobility Grants P2GEP2\_191444  and P2GEP2 194800 respectively. We acknowledge support from the Government of Spain (FIS2020-TRANQI and Severo Ochoa CEX2019-000910-S), Fundacio Cellex, Fundacio Mir-Puig, Generalitat de Catalunya (CERCA, AGAUR SGR 1381 and QuantumCAT), the ERC AdG CERQUTE, the AXA Chair in Quantum Information Science and the Swiss NCCR SwissMap.

\section*{Published version}
The reader can find the published version of this pre-print in \url{https://www.nature.com/articles/s41586-021-04160-4}.

\bibliography{rvsc_bib}

\newpage

\begin{appendix}

In this Appendix, we formally state and prove the results mentioned in the main text. It is organized as follows. In section \ref{sec:tensorAQFT}, we discuss the status of postulate $(iv)$ in algebraic quantum field theory. In section~\ref{sec:realtheories}, we describe existing real quasi-quantum theories that have the same predictive power as complex quantum theory. In section \ref{sec:non-exp}, we counter two common theoretical arguments against real quantum theory. In section~\ref{sec:jointmeasindeppreparations}, we show that independently prepared states jointly measured in a single location cannot rule out real quantum theory. In section \ref{sec:setup} we introduce the considered experimental scenario and state our main results, that we prove in sections \ref{sec:pbar}, \ref{sec:noisy} and \ref{numerics}. In section \ref{sec:experimental} we discuss the challenges and assumptions of a future experimental implementation.

\section{Tensor products in algebraic quantum field theory}\label{sec:tensorAQFT}

%\marco{I think there is a need of a sentence or two motivating the problem discussed in this section. something more or less like:
%The postulate $(iv)$ seems to be incompatible with quantum field theory, where the wave function of a fermion needs to be anti-symmetrized. 
%In this section, we show that this postulate still holds in this theory, assuming the \emph{split property} defined below. 
%}

Postulate (iv) is generally believed not to apply in quantum field theory, where space-like separation is usually modeled through commutation relations rather than tensor products. In this section, we show that, contrary to this widely-held opinion, postulate (iv) is a necessary condition for quantum field theories to be \emph{physical}.

Algebraic quantum field theory (AQFT) is an axiomatization of quantum field theory developed by Haag, Kastler, Araki and others in the 1950s \cite{book_Haag}. The starting point of AQFT is a space-time manifold $M$ and an overall Hilbert space $\H$. For any space-time bounded region $O\subset M$, there exists a $C^*$-algebra $A(O)$, acting on $\H$, whose elements correspond to the set of measurements and operations that an experimenter acting in $O$ can conduct. Several postulates try to capture the properties of these algebras if they are to represent the local observables of a reasonable quantum field theory. We will just mention one: If $O_1,O_2\subset M$ are causally unconnected, then $[A(O_1),A(O_2)]=0$. This condition, called microcausality, formalizes the notion that experimenters acting in causally unconnected regions can conduct operations independently of one another. 

In a quantum field theory we do not only expect to be able to conduct independent measurements, but also independent local state preparations. In the language of AQFT, a local state $\phi$ in the region $O\subset M$ corresponds to a linear functional $\phi:A(O)\to \mathbb{C}$ such that $\phi(xx^\dagger)\geq 0$, for all $x\in A(O)$, and $\phi(1)=1$. The state is called \emph{normal} if there exists a trace-class positive semidefinite operator $\rho$ acting on $\H$ such that $\phi(x)=\tr\rho x$ for all $x\in A(O)$. Intuitively, in order to prepare a local state in a region $O$, we might need to act on a slightly larger region $O'$ (e.g.: to shield $O$ from cosmic rays). Denoting by $B(\H)$ the set of bounded operators acting on $\H$, a local preparation of $\phi$ is thus defined as a completely positive map $T:B(\H)\to B(\H)$ such that $T(x)=\phi(x)T(1)$, for all $x\in A(O)$ and $T(y')=T(1)y'$, for all $y'$ commuting with $A(O')$, where $O'\in M$ represents a space-time region that strictly contains $O$. Note that, in principle, $T(1)\not=1$, i.e., we allow for the preparation to be non-deterministic.

As shown by Werner \cite{Werner}, the possibility of conducting, for any bounded $O\in M$, a non-deterministic local preparation of some normal state $\phi$ implies that one can prepare any normal state in $O$ deterministically by means of maps $T$ of the form $T(\omega)=\sum_j c_j\omega c_j^\dagger$, with $c_j\in A(O')$. In \cite{Werner}, this condition is also proven equivalent to the so-called \emph{split property}.
The split property demands, for every $O, O'\subset M$ with $O$ strictly contained in $O'$, the existence of a type-I von Neumann factor (namely, a $C^*$-algebra isomorphic to $B({\cal K})$ for some Hilbert space ${\cal K}$) $A$ such that $A(O)\subset A\subset A(O')$. Due to the previous considerations and the fact that large classes of solvable models satisfy it \cite{models_AQFT, models_AQFT2}, the split property is usually regarded as an extra postulate of AQFT, when it is not derived from general assumptions of thermodynamic stability \cite{nuclearity}.

How does the split property relate to experiments with space-like separated experimenters? Type-I factors $A$ have the convenient property of factoring the Hilbert space $\H$ where they act. Namely, if $A$ is a type-I factor, then there exist Hilbert spaces ${\cal K}$, ${\cal K}'$ and a unitary $U:\H\to {\cal K}\otimes {\cal K}'$ such that $UAU^\dagger =B({\cal K})\otimes\id_{K'}$ and $UA'U^\dagger =\id_K\otimes B({\cal K}')$, where the commutant $A'$ is  the set of operators in $B(\H)$ that commute with all the elements of $A$.

Now, consider a scenario where two parties Alice and Bob conduct measurements in a space-like separated way. Call $O_A$ ($O_B$) the space-time region where Alice (Bob) conducts her (his) experiments, and suppose that there exists a region $O_A'$ that strictly contains $O_A$ and such that $O'_A,O_B$ are space-like separated. Then, by virtue of the split property, there exists a type-I algebra $A$ such that $A(O_A)\subset A\subset A(O_A')$. Since by microcausality $A(O_B)\subset A(O_A')'\subset A'$, it follows that there exists a unitary $U:\H\to {\cal K}\otimes {\cal K}'$ such that $UA(O_A)U^\dagger =\tilde{O}_A\otimes \id_{{\cal K}'}$, $UA(O_B)U^\dagger =\id_{{\cal K}}\otimes \tilde{O}_B$. That is, the interactions of each party can be understood to take place in different factors of an overall Hilbert space. This argument generalizes for arbitrarily many parties. 

In multi-partite scenarios in AQFT, the overall Hilbert space can thus be decomposed as the tensor product of multiple factors, one for each party, and the physical operations carried out by each party are represented by operators acting trivially on all the other factors. In this framework, if each party conducts a local state preparation (as defined above), the statistics observed by all the parties can be seen to correspond to those generated by a product quantum state. Hence postulate (iv) in the main text also holds in quantum field theories satisfying the split property.

\section{Reproducing the predictions of quantum mechanics through non-local real theories}\label{sec:realtheories}

Following the development of quantum theory, alternative theories that retain only part of the quantum formalism were developed, usually with the aim to enforce certain desirable properties. 

In his 1960 seminal paper \cite{stueckelberg}, Stueckelberg tackled the problem of constructing a version of quantum theory that does not require complex numbers. His construction is mathematically equivalent to the one that Regina proposes in the main text to explain Conan's single-site experiments. In his work \cite{stueckelberg}, Stueckelberg is just concerned with single-site experiments, so he does not provide a prescription to describe multipartite experiments. 

A straightforward generalization of Stueckelberg's theory, advocated in \cite{Aleksandrova_2013}, consists in positing the existence of a universal quantum bit, partially accessible to each observer. In this theory, the Hilbert space $\H$ of any quantum system factors into as many Hilbert spaces as independent subsystems, plus the universal qubit $U$. The operations that an experimenter with access to factor $A$ can perform are of the form 

\be
\proj{i}_U\otimes O_A\otimes\id_{\tilde{A}}+\proj{-i}_U\otimes O^*_A\otimes \id_{\tilde{A}}, 
\ee
\noindent where $\tilde{A}$ denotes all factors decomposing $\H$ other than $A$ and $U$. This generalization recovers all quantum predictions. If one identifies the basic elements of this theory (namely, vectors and projectors) with those of quantum theory, then one finds that Stueckelberg's theory violates postulate $(iv)$ in the main text, as all subsystems are allowed to operate on factor $U$. One can thus interpret Stueckelberg's theory as a non-local (yet non-signalling) real variant of quantum theory.

Another alternative that deserves specific mention, for its similarities to Stueckelberg's theory, is Bohmian mechanics~\cite{Bohm}. Here, in addition to the quantum wave function, particle positions are part of the ontology. Even though it is generally formulated over complex Hilbert spaces, Bohmian mechanics can incorporate the above construction using a universal qubit. Violating postulate $(iv)$ may be considered a less serious flaw in this case, as the wave function is not associated with the states of individual particles but rather guides their dynamics in a fundamentally non-local way anyway. No wonder, this theory is usually regarded as a non-local (yet non-signalling) classical theory.

%Overall, our results imply that all real quasi-quantum theories one may come up with will necessarily suffer from similar shortcomings and violate at least one of the four postulates $(i', ii-iv)$.

\section{Non-experimental arguments against real quantum theory}\label{sec:non-exp}

Our work proves that real quantum theory, i.e., the theory satisfying the postulates $(i_\R),(ii)-(iv)$, cannot reproduce the predictions of complex quantum theory in the entanglement swapping scenario. 
Prior to our work, some arguments against real quantum theory, also called real quantum mechanics \cite{real_quantum1, real_quantum2}, have appeared in the quantum foundations literature. As we explain below, such arguments are not associated to any experimental observation: rather, they amount to non-falsifiable conceptual considerations related to the internal properties of the theory. 

A natural argument against a real quantum theory is based on dimension considerations. For instance, photon polarization is described by a complex Hilbert space of dimension $2$, hence is determined by three independent parameters $n_x,n_y,n_z$. In contrast, if the Hilbert space of photon polarization were real, then just two numbers $n_x, n_z$ suffice to describe a physical state. The successful encoding and decoding of three parameters in identical preparations of a photon state seems to be a solid argument against real quantum theory. 

Note, however, that Hilbert space dimension cannot be experimentally lower bounded. Interpreting such an experiment as a disproof of real quantum theory thus requires a considerable leap of faith. In fact, real quantum theory can account for the measurement statistics of any such experiment by postulating that photon polarization lives in a $4$-dimensional real Hilbert space. According to such a real quantum interpretation of photon experiments, photon polarization would allow storing, not just three, but nine continuous parameters (as such is the number of independent components of a real $4\times 4$ normalized density matrix). The fact that we do not seem to have access to those six extra degrees of freedom can be explained away by appealing to some super-selection rules, or more simply to our own limitations as experimental physicists. 

%More fundamentally, remark that all experimental systems are infinite dimensional in our current understanding of the world. Finite dimensionality can be seen as a theoretical simplification, as reflected by the postulates of quantum theory $(i-iv)$, which are phrased in terms of a priori infinite Hilbert spaces.\mirjam{I would say that already using Hilbert spaces to describe reality is a theoretical simplification, so somehow the infinite dimensionality of the real system and of the Hilbert space are already 2 different things. I'm therefore not sure the part from ``More fundamentally'' on helps for what we want to say?}\marco{I agree. I wanted to develop what we say in the main, but his is only rephrasing, not at the best place. Unless you see a good way to do this, delete it?}

A second argument often heard against real quantum theory is the fact that it violates the principle of local tomography: namely, there exist distinct multiparty real quantum states which cannot be distinguished via local measurements. Local tomography has been introduced in several axiomatizations of quantum mechanics, see e.g.~\cite{Hardy_axioms, Chiribella_axioms}. However, whether this property should be regarded as a reasonable physical principle is controversial, as evidenced by the interest in physical models that clearly violate it \cite{real_quantum1, real_quantum2, non_local_tomography}. Note also that, like the dimensional arguments, any observed violation of local tomography can be explained away by arguing that the local measurements conducted on each subsystem were not exhaustive. 

Other abstract mathematical arguments against (or in favor of)  real quantum theory exists, based on symmetry considerations~\cite{Moretti_2017, Baez_2011, Dyson_1962}. None are experimentally falsifiable.

\section{Real simulation of joint measurements on independent preparations}\label{sec:jointmeasindeppreparations}
There exist no-go theorems in quantum physics, such as the Pusey-Barrett-Rudolph (PBR) theorem~\cite{PBR}, that deal with a scenario in which $N$ independent sources prepare quantum states that are sent to a central node, where joint measurements are performed. Our goal here is to provide a strategy for the simulation using real quantum theory of any such experiment in complex quantum theory. In what follows, we restrict the analysis to the case in which each source prepares $P$ possible states and only one measurement is performed, the generalization to more  measurements being straightforward.

In the considered scenario, each source $i$, with $i=1,\ldots,N$, prepares the complex states $\{\rho^{(i)}_{p_i}\}_{p_i=1,\ldots,P}$, which are sent to a central node that performs on the $N$ states the measurement of $R$ possible results defined by the complex operators $\{M_r\}_{r=1,\ldots,R}$. The obtained statistics is described by the conditional probability distribution
\begin{equation}
    P(r|p_1\ldots p_N)=\tr{\left(\rho^{(1)}_{p_1}\otimes \rho^{(2)}_{p_2}\otimes\ldots \otimes \rho^{(N)}_{p_N} M_r\right)} .
\end{equation}
A possible real simulation works as follows. The preparations by each source $i$ are just encoded in a basis of a real Hilbert space of dimension $P$, $\{\ket{p_i}\}_{p_i=1,\ldots,P}$. The measurement operators are defined by the positive operators
\begin{equation}
    M_r=\sum_{p_1,\ldots,p_N}P(r|p_1\ldots p_N)\proj{p_1\ldots p_N} ,
\end{equation}
which sum up to the identity. It is simple to see that 
\begin{equation}
    P(r|p_1\ldots p_N)=\tr{\left(\proj{p_1\ldots p_N} M_r\right)} .
\end{equation}

Note that this construction provides an alternative real simulation of any experiment involving a single quantum system different from the one presented in Figure~\ref{fig:scenarios}. We were unable to adapt the construction with the extra qubit in this Figure to the independent preparation scenario considered here. The impossibility comes from the fact that at the measuring device complex conjugation, a positive but not completely positive map, needs to be applied to parts of the system. In fact, based on numerical simulations, we conjecture that there are experiments involving joint measurements on independent preparations in complex Hilbert spaces of dimension $d$ that require for its real simulation a Hilbert space dimension strictly larger than $2d$. On the other hand, the simple encoding of the preparations and measurements in a real basis discussed here does not work in experiments displaying a Bell inequality violation, as they require entanglement, hence quantum coherence. While one should never forget that these are just possible real simulations, not necessarily unique, the conflict appearing between the two approaches gives another intuitive explanation of why a real simulation becomes impossible when combining independent preparations and measurements.

\section{Our setup and main results} \label{sec:setup}

Causal networks are the natural generalizations of Bell scenarios to a richer class of causality relations. In this work, we consider the causal network depicted in Fig. \ref{fig:entswap} (lower pane), which is similar to an entanglement swapping setup and therefore called the \emph{SWAP scenario}. A (shared randomness) source $\lambda$ controls two spacelike-separated sources $S_L,S_R$ (for left and right source) each of which further generates a bipartite system to be sent towards three distant observers Alice, Bob and Charlie. The observers can choose measurement $x,y,z$ to conduct on the incoming systems locally, and receive an outcome $a,b,c$. Repetition of the experiment allows them to collect a family of joint probability distributions $P(a,b,c|x,y,z)$. Going forward, in our SWAP scenario, Alice and Charlie have, respectively, three, and six dichotomic measurement settings, i.e., $x=1,2,3; z=1,...,6$ and $a,c=-1,1$, while Bob's only measurement is assumed to have four outcomes $b=b_1b_2=00,01,10,11$, that we equivalently label as $b=\phi^+,\psi^+,\phi^-,\psi^-$ for reasons that will soon be obvious.

Just like in Bell scenarios, the source $\lambda$ is assumed to be a {\em classical} source modelled by a classical probability distribution $P(\lambda)$. Moreover, the causal network forbids the possibility that $\lambda$ influences $x,y,z$ (also known as {\em measurement dependence}). The reason for classicality is that $S_L,S_R$ may be produced in the same factory or being operated using the same power socket, which are believed to be sources of classical correlations. Furthermore, the causal network with an additional source $\lambda$ is clearly more general than the one without such source, and in fact our results even hold in this setting.

%Alice and Bob are distributed the quantum state $\tilde{\sigma}_{AB_1}$, while Bob and Charlie are distributed the state $\tilde{\sigma}_{B_2C}$. These states might have been prepared using some shared randomness $\lambda$.
 %This shared randomness is especially relevant when working towards an experimental implementation of our results (see Theorem~\ref{main_theo} below). For that purpose, it is natural to assume that the two independent sources have some (classical) common origin, e.g.\ the production in the same factory or being operated by the same researchers, which may introduce unwanted and undetected classical correlations. Nevertheless, allowing for this also strengthens our results from a theoretical viewpoint: we will show below that, even if we have access to shared randomness in real quantum theory, we still cannot reproduce certain quantum correlations (that themselves do not use any shared randomness).

%Alice and Charlie can choose which measurement $x$ or $z$ to conduct, with results $a, c$. Bob, on the other hand, is just allowed to conduct a single measurement, with outcome $b$. This causal network is similar to an entanglement swapping setup and we therefore call it the \emph{SWAP scenario}. 

Interpretation within complex complex quantum theory implies that $P(a,b,c|x,z)$ must be of the form

\be
P(a,b,c|x,z)=\sum_\lambda P(\lambda)\tr\left\{(\tilde{\sigma}^{\lambda}_{AB_1}\otimes\tilde{\sigma}^{\lambda}_{B_2C})(\tilde{A}_{a|x}\otimes \tilde{B}_{b}\otimes \tilde{C}_{c|z})\right\},
\label{decomp}
\ee
\noindent for some quantum states $\tilde{\sigma}_{AB_1}^\lambda,\tilde{\sigma}_{B_2C}^\lambda$, some probability distribution $P(\lambda)$ and projective measurement operators $\tilde{A}_{a|x},\tilde{B}_{b},\tilde{C}_{c|z}$, with $\sum_{a}\tilde{A}_{a|x}=\id_A, \sum_b \tilde{B}_b=\id_B, \sum_c\tilde{C}_{c|z}=\id_C$.
 Interpreting the SWAP scenario according to {\em real} quantum theory gives exactly the same equation but with all operators  restricted to acting on a {\em real} Hilbert space or, equivalently, having real matrix entries. 
 
 Note that, rather than projective measurements, Alice, Bob and Charlie could in principle conduct the more general Positive Operator valued Measures (POVMs), whereby the operators $\tilde{A}_{a|x}, \tilde{B}_b, \tilde{C}_{c|z}$ are just required to be positive semidefinite instead of projectors. However, any correlation $P(a,b,c|x,z)$ admitting a decomposition of the form (\ref{decomp}), with $\tilde{A}_{a|x}, \tilde{B}_b, \tilde{C}_{c|z}$ denoting (real) POVM elements, can also be reproduced with (real) projective measurements. This is so because, through Naimark dilations, one express any local (real) POVM as a local (real) projective measurement acting over the original quantum system and a local ancillary system in a (real) pure state. Hence, in order to disprove that a given distribution $P(a,b,c|x,z)$ admits a real quantum representation (\ref{decomp}), it suffices to do so under the assumption of projective measurements.
 
We now describe the specific correlation $\bar{P}(a,b,c|x,z)$ and Bell inequality $\mathscr{T}$ that witness the separation between real and complex quantum theory. Let the two quantum sources distribute the states $\bar{\sigma}_{AB_1}=\bar{\sigma}_{B_2C}=\Phi^+=\proj{\phi^+}$, where $\ket{\phi^+}=\frac{1}{\sqrt{2}}(\ket{00}+\ket{11})$. Alice's three dichotomic observables $(\bar{A}_{1|x}-\bar{A}_{-1|x}:x=1,2,3)$ correspond to the three Pauli measurements $\sigma_Z,\sigma_X,\sigma_Y$, and Charlie measures the dichotomic operators
\be
\bar{D}_{ij}=\frac{\sigma_i+\sigma_j}{\sqrt{2}}, \qquad \bar{E}_{ij}=\frac{\sigma_i-\sigma_j}{\sqrt{2}},
\ee
\noindent for $ij=zx,zy,xy$, which correspond to the observables $(\bar{C}_{1|z}-\bar{C}_{-1|z}:z=1,\ldots,6)$ when taken in the order $D_{zx}$, $E_{zx}$, $D_{zy}$, $E_{zy}$, $D_{xy}$, $E_{xy}$. Bob conducts a Bell basis measurement, with outcomes $b$ corresponding to each of the orthogonal projections onto the states $\ket{\phi^{\pm}}=\frac{1}{\sqrt{2}}(\ket{00}\pm\ket{11})$, $\ket{\psi^{\pm}}=\frac{1}{\sqrt{2}}(\ket{10}\pm\ket{01})$. We warn the reader that our convention of $\ket{\psi^-}$ differs from the usual one so as to make our formulas simpler. 

These define the resulting distribution $\bar{P}=\{\bar{P}(a,b,c|x,z):a,b,c,x,z\}$, which obviously admits a representation of the form (\ref{decomp}), with \emph{complex} measurement operators. Since $\bar{P}$ does not require shared randomness $\lambda$ for its realization, it is also compatible with the less general causal network depicted in Figure \ref{fig:entswap} (upper pane) in the main text.

Given some distribution $P(a,b,c|x,z)$, define $S^b_{xz}=\sum_{a,c=-1,1}P(a,b,c|x,z)ac$ and the linear functional

\begin{align}
\mathscr{T}_b(P)=& (-1)^{b_2}(S_{11}^b+S^b_{12})+(-1)^{b_1}(S^b_{21}-S^b_{22})+\nonumber\\
&(-1)^{b_2}(S^b_{13}+S^b_{14})-(-1)^{b_1+b_2}(S^b_{33}-S^b_{34})+\nonumber\\
& (-1)^{b_1}(S^b_{25}+S^b_{26})-(-1)^{b_1+b_2}(S^b_{35}-S^b_{36}).
\label{adjusted_3CHSH}
\end{align}
Note that $S^b_{xz}$ is the same as the conditional expectation value $\langle A_xC_z\rangle$ when conditioned on the event that Bob receives the outcome $b$. Therefore, for $b=00$ then $\mathscr{T}_{00}(P)$ coincides with the LHS of the $\text{CHSH}_3$ Bell inequality for Alice and Charlie, as defined in eq. (\ref{eq:chsh3}), but evaluated for $P(a,b,c|x,z)$; for the other $b$ it corresponds to variations of the same inequality.
It can be verified that, for the considered states and measurements, the identity $\mathscr{T}_b(\bar{P})=6\sqrt{2}\bar{P}(b)$ holds for all $b$, with $\bar{P}(b)=\frac{1}{4}$. This means that, conditioned on any measurement outcome $b$ of Bob's, the state of Alice and Charlie allows them to maximally violate one of the variants of the $\text{CHSH}_3$ Bell inequality, namely the one corresponding to $\mathscr{T}_b$.

Our first result, proven in Section~\ref{sec:pbar}, is the impossibility to reproduce $\bar{P}$ exactly using real quantum physics.

\begin{prop}
\label{exact_case}
$\bar{P}$ does not admit a decomposition of the form (\ref{decomp}) if we demand the states $\tilde{\sigma}_{AB_1}^\lambda,\tilde{\sigma}_{B_2C}^\lambda$ and measurement operators $\tilde{A}_{a|x},\tilde{B}_{b},\tilde{C}_{c|z}$ to be real, regardless of the dimension of the underlying real Hilbert space.
\end{prop}

A more elaborate argument, presented in section \ref{sec:noisy}, leads to the following robustness claim.

\begin{theo}
\label{approximate_case}
Let $P(a,b,c|x,z)$ be a distribution such that $|P(b)-\frac{1}{4}|< \varepsilon_c$ and $\mathscr{T}_b(P)>(6\sqrt{2}-\varepsilon_c)P(b)$, for all $b$, with $\varepsilon_c\approx 7.18 \cdot 10^{-5}$. Then, $P(a,b,c|x,z)$ does not admit a decomposition of the form (\ref{decomp}), with real states and measurement operators, regardless of the dimension of the underlying real Hilbert space.
\end{theo}

Maximally violating the Bell inequality $\text{CHSH}_3$ with a precision of order $\varepsilon_c$ right after a Bell measurement is way beyond the capabilities of current quantum technologies. Define $\mathscr{T}(P)=\sum_{b\in\{0,1\}^2}\mathscr{T}_b(P)$, and note that in the considered setup, $\mathscr{T}(\bar{P})=6\sqrt{2}\approx 8.4852$. 
In section \ref{numerics} we use ideas from non-commutative polynomial optimization \cite{NOP} to obtain the experimentally friendly robustness bound.

\begin{theo}
\label{main_theo}
For any distribution $P$ admitting a decomposition of the form (\ref{decomp}) with real quantum states and real measurement operators acting on real Hilbert space of arbitrary dimension, 
\be
\mathscr{T}(P)=\sum_{b\in\{0,1\}^2}\mathscr{T}_b(P)\leq 7.6605.
\label{violation}
\ee
\end{theo}

\noindent The figure appearing in the theorem is the solution of a semidefinite program (SDP) \cite{sdp}, a type of convex optimization problem that can be solved in polynomial time. To arrive at this result, we used the SDP solver Mosek \cite{mosek}, with the MATLAB package YALMIP \cite{yalmip}. Note that, from duality theory, we can \emph{certify} that the solution of the SDP is the one stated in the theorem, up to computer precision. That is, the Theorem above represents a rigorous mathematical result, and can be understood as a computer-generated proof.

\section{Proof of Proposition \ref{exact_case}} \label{sec:pbar}

We prove the result by contradiction. Suppose that, indeed, there exist a distribution $P(\lambda)$, Hilbert spaces $A,B_1,B_2,C$, real states $\tilde{\sigma}_{AB_1}^\lambda,\tilde{\sigma}_{B_2C}^\lambda$ and real measurement operators $\tilde{A}_{a|x}, \tilde{B}_{b}, \tilde{C}_{c|z}$ such that eq.~\eqref{decomp} holds with $P(a,b,c|x,z)=\bar{P}(a,b,c|x,z)$. Let $\psi:=\sum_\lambda P(\lambda)\tilde{\sigma}^{\lambda}_{AB_1}\otimes\tilde{\sigma}^{\lambda}_{B_2C}$ be the global state for systems $A, B_1, B_2, C$ at the beginning of the experiment. With probability $P(b)=\tr(\psi \tilde{B}_b)$, Bob observes outcome $b$ and collapses the $AC$ system into the state $\tr_{B}(\tilde{B}_b \psi \tilde{B}_b)/P(b)$.

Going forward, consider a real purification $\ket{\psi^b}_{ACP}$ of Alice and Charlie's conditional state, and redefine Charlie's system $C$ to accommodate the purifying system $P$, over which Charlie's measurement operators act trivially. This simplifies our notation by replacing density matrices on $AC$ by pure states on $A(CP)$. Most importantly, the choice of purification does not impact our result: this is so because the only role that $\ket{\psi^b}_{ACP}$ plays in the proof consists in being acted on system $A$ and the original system $C$ with linear maps, after which the purifying system $P$ is traced out, together with systems $A,C$. Hence, assuming the existence of such a purification does not affect the final expressions and hence the validity of the proof. Consequently, by a slight abuse of notation, we will from now on refer to the system $CP$ simply as $C$. 

We denote Alice's three dichotomic observables by $Z^A, X^A, Y^A$, referring to $(\tilde{A}_{1|x}-\tilde{A}_{-1|x}:x=1,2,3)$. Analogously, Charlie's dichotomic observables $(\tilde{C}_{1|z}-\tilde{C}_{-1|z}:z=1,...,6)$ are denoted by $D^C_{zx}, E^C_{zx}$, $D^C_{zy}$, $E^C_{zy}$, $D^C_{xy}$, $E^C_{xy}$, respectively. 
We mostly work with the operator version of \eqref{adjusted_3CHSH}, which in this notation takes the form
\begin{align}
    \hat{\mathscr{T}}_b =& (-1)^{b_2} Z^A(D^C_{zx}+E^C_{zx}) + (-1)^{b_1}X^A(D^C_{zx} - E^C_{zx}) \nonumber\\
    &(-1)^{b_2} Z^A(D^C_{zy} + E^C_{zy}) - (-1)^{b_1+b_2} Y^A (D^C_{zy} - E^C_{zy}) \nonumber \\
    &(-1)^{b_1} X^A(D^C_{xy} + E^C_{xy}) - (-1)^{b_1+b_2} Y^A (D^C_{xy} - E^C_{xy}).
\end{align}
The fact that we start with real projective measurements translates into these nine operators satisfying $O^2=\id$ and $O^T=O$.

\begin{figure}
\begin{center}
\centerline{
\Qcircuit @C=1em @R=1em{
 \lstick{A'} &\qw  & \gate{H} & \ctrl{2} \ar@{.}[]+<1.6em,1em>;[d]+<1.6em,-10.5em>    & \gate{H} & \ctrl{2} \ar@{.}[]+<1.6em,1em>;[d]+<1.6em,-10.5em> & \qw \ar@{.}[]+<2.3em,1em>;[d]+<2.3em,-10.5em>      & \qw  &  \qw\\
 \lstick{A''}& \qw & \gate{H} & \qw         & \qw      & \qw        & \ctrl{1}     & \gate{H} & \qw \\
 \lstick{A}  & \qw & \qw      & \gate{Z^A} & \qw       & \gate{X^A} & \gate{Y^AX^A} & \qw & \qw \\
 \lstick{C}  & \qw & \qw      & \gate{\hat{Z}^C} & \qw       & \gate{\hat{X}^C} & \gate{\hat{Y}^C\hat{X}^C} & \qw & \qw \\
 \lstick{C''}& \qw & \gate{H} & \qw         & \qw      & \qw        & \ctrl{-1}     & \gate{H} & \qw \\
 \lstick{C'} &\qw  & \gate{H} & \ctrl{-2}    & \gate{H} & \ctrl{-2}  & \qw            & \qw &  \qw\\
 & & & & & & & & \\
 & & {\text{Step 1}} & & { \qquad \ \text{Step 2}} & &  { \text{Step 3}} &   {\quad \ \text{Step 4}} & \\}}
\end{center}
\caption{Real local isometry $U\otimes V$ built from each party's untrusted measurement operators. For later reference, we denote the local isometry performing Steps 1 and 2 of the circuit as $U'\otimes V'$, and the remaining gates (performing the operations of Steps 3 and 4) as $U''\otimes V''$. The operators $\hat{X}^C$, $\hat{Y}^C$, $\hat{Z}^C$ are defined in terms of Charlie's measurement operators as the regularised versions of $X^C=\frac{D^C_{zx}-E^C_{zx}}{\sqrt{2}}$, $Y^C=\frac{D^C_{zy}-E^C_{zy}}{\sqrt{2}}$ and  $Z^C=\frac{D^C_{zx}+E^C_{zx}}{\sqrt{2}}$ respectively. For any operator $O$ let $\hat{O}:=O_{0\to1}\abs{O_{0\to1}}^{-1}$ where $O_{0\to1}$ is the resulting operator after setting any zero eigenvalues of $O$ to one. This regularization procedure~\cite{self_testing} turns a hermitian operator $O$ into a unitary operator with eigenvalues $\pm1$. Finally, $H$ denotes the Hadamard gate, defined through the relations $H\ket{0}=\frac{1}{\sqrt{2}}(\ket{0}+\ket{1})$, $H\ket{1}=\frac{1}{\sqrt{2}}(\ket{0}-\ket{1})$.
}
\label{isometry_pic}
\end{figure}
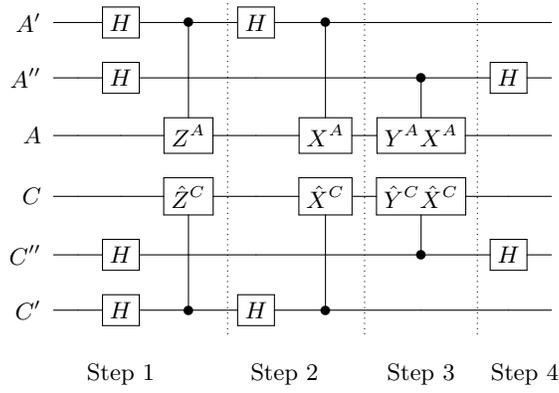

Taking inspiration from \cite{self_testing}, we consider the real isometry depicted in Figure \ref{isometry_pic}. Our first claim is that, if the real state $\ket{\psi^b}$ saturates the quantum bound $\bra{\psi^b}\hat{\mathscr{T}}_b\ket{\psi^b}\leq 6\sqrt{2}$ for $b=\phi^+,\psi^+,\phi^-,\psi^-$, then the effect of this isometry on systems $A'A''C'C''$ is to prepare the states
\begin{align} \label{eq:perfect_state}
\rho^b:= \tr_{AC}({U}\otimes {V} \ketbra{\psi^b}{\psi^b} {U}^\dagger\otimes {V}^\dagger)=\ketbra{b}{b}_{A'C'}\otimes\left[\frac{\Psi^++\Phi^-}{2}\right]_{A''C''}.
\end{align}
Note that, as promised, the isometry only acts on system $A$ and the original system $C$ (since it is built from Charlie's measurement operators). Systems $A$ and $CP$ are traced out; hence the same final state $\rho^b$ would have been obtained had we not purified Alice and Charlie's post-selected state. Now we prove this equation for the case $b=\phi^+$ only, noting that the other cases are completely analogous (or can be checked in Appendix~\ref{sec:noisy}). We explicitly show the effect of the isometry step by step. 
\begin{itemize}
\item After steps $1$ and $2$ of the isometry we are left with the state
\[
\ket{\psi^{\phi^+}}_{AC} \ket{0000}_{A'A''C'C''} \mapsto \ket{\phi^+}_{A'C'} \ket{+\!+}_{A''C''} \frac{\id + Z^A}{\sqrt{2}} \ket{\psi^{\phi^+}}_{{AC}}. 
\]
This follows from eq. (A12) in \cite{self_testing}, in turn a consequence of the maximal violation of the CHSH inequality through Alice and Charlie's first two measurement settings. We can now forget about systems $A'C'$, since the subsequent operations defining the isometry do not act on them. 

\item After Step $3$, systems ${A}A''{C}C''$ are in state
\begin{equation} \label{eq: middle state}
\frac{1}{2}(\ket{00}_{A''C''} \ket{\varphi}_{{AC}} + \ket{11}_{A''C''} Y^AX^A\hat{Y}^C \hat{X}^C \ket{\varphi}_{{AC}} + \ket{01}_{A''C''} \hat{Y}^C \hat{X}^C \ket{\varphi}_{{AC}}+\ket{10}_{A''C''} Y^AX^A \ket{\varphi}_{{AC}}),
\end{equation}
where $\ket{\varphi}_{{AC}}= \frac{\id+Z^A}{\sqrt{2}}\ket{\psi^{\phi^+}}_{{AC}}$. 
Since $\ket{\psi^{\phi^+}}$ maximally violates a $\text{CHSH}_3$ inequality, the identities 

\begin{align}
&\{X^A,Y^A\}_{+}\ket{\psi^{\phi^+}}_{{AC}}=\{X^A,Z^A\}_{+}\ket{\psi^{\phi^+}}_{{AC}}=\{Z^A,Y^A\}_{+}\ket{\psi^{\phi^+}}_{{AC}}=0,\nonumber\\
&\hat{Z}^C \ket{\psi^{\phi^+}}_{{AC}} = Z^A \ket{\psi^{\phi^+}}_{{AC}},\nonumber\\ 
&\hat{Y}^C \ket{\psi^{\phi^+}}_{{AC}} = -Y^A \ket{\psi^{\phi^+}}_{{AC}},\nonumber\\
&\hat{X}^C \ket{\psi^{\phi^+}}_{{AC}} = X^A \ket{\psi^{\phi^+}}_{{AC}}
\end{align}
\noindent hold, see \cite{self_testing} or Appendix F for a detailed proof. In turn, those imply the relations $\hat{Y}^C\hat{X}^C \ket{\varphi}_{{AC}} = Y^AX^A \ket{\varphi}_{{AC}}$ and $Y^A X^A\hat{Y}^C\hat{X}^C \ket{\varphi}_{{AC}} = - \ket{\varphi}_{{AC}}$.
Therefore \eqref{eq: middle state} is the same as 
\[
\frac{1}{\sqrt{2}} (\ket{\phi^-}_{A''C''} \ket{\varphi}_{{AC}} + \ket{\psi^+}_{A''C''} Y^A X^A \ket{\varphi}_{{AC}}).
\]

\item The Hadamard gates of step $4$ change $\ket{\phi^-}_{A''C''}$ to $\ket{\psi^+}_{A''C''}$ and viceversa. The final state is therefore
\[
\frac{1}{\sqrt{2}}\ket{\phi^+}_{A'C'}(\ket{\psi^+}_{A''C''} \ket{\varphi}_{{AC}} + \ket{\phi^-}_{A''C''} Y^A X^A \ket{\varphi}_{{AC}}).
\]
Finally, we take the partial trace over systems ${AC}$ to obtain
\[
\frac{1}{2} \Phi^+ \otimes \left(\bra{\varphi}Y^A X^A \ket{\varphi}(\ketbra{\phi^-}{\psi^+}+\ketbra{\psi^+}{\phi^-}) + \braket{\varphi}{\varphi} (\Psi^++\Phi^-)\right).
\]
Since $U\otimes V$ is an isometry, it immediately follows that $\braket{\varphi}{\varphi} = 1 $, which is equivalent to saying that $\bra{\psi^{\phi^+}}Z^A \ket{\psi^{\phi^+}}=0$. So far, all our considerations are general and did not use the fact that we are dealing with real numbers. Taking this fact into account, we have that
\[
\bra{\varphi} Y^A X^A \ket{\varphi} = \overline{\bra{\varphi} Y^A X^A \ket{\varphi}} = \bra{\varphi} X^A Y^A  \ket{\varphi} = - \bra{\varphi} Y^A X^A \ket{\varphi}.
\]
Therefore, this quantity is zero and the state left in systems $A'C'A''C''$ is indeed $\rho^b$.
\end{itemize}

To reach a contradiction, note that
\begin{equation}
    \frac{\Psi^+ +\Phi^-}{2}=\frac{\proj{i}^{\otimes 2}+\proj{-i}^{\otimes 2}}{2}
\end{equation}
%. This shows that the real construction sketched in Figure \ref{fig:scenarios} (right pane) is somehow unavoidable: in order to simulate the statistics of Alice and Charlie (conditioned on Bob's outcome $b$), one needs to prepare the state $\frac{1}{2}(\proj{b}\otimes \proj{i}^{\otimes 2}+\proj{b}^*\otimes \proj{-i}^{\otimes 2})=\proj{b}\otimes \frac{\proj{i}^{\otimes 2}+\proj{-i}^{\otimes 2}}{2}$ (up to local isometries), as the construction requires.
and therefore summing over Bob's results, the application of $U\otimes V$ leaves system $A'A''C'C''$ in the state
\be
\rho=\sum_b \bar{P}(b)\rho^b=\frac{1}{4}\left(\rho^{\phi^+}+\rho^{\phi^-}+\rho^{\psi^+}+\rho^{\psi^-}\right)=\frac{\id_{A'C'}}{4}\otimes \left[\frac{\proj{i}^{\otimes 2}+\proj{-i}^{\otimes 2}}{2}\right]_{A''C''}.
\label{suma}
\ee
On the other hand, from (\ref{decomp}) we have that the same application of $U\otimes V$ should give the state
\be
\rho=\sum_\lambda P(\lambda) \tr_A(U\tilde{\sigma}^\lambda_A U^\dagger)\otimes\tr_C(V\tilde{\sigma}^\lambda_C V^\dagger).
\ee
\noindent The last two equations contradict each other. The last state is always a real separable state with respect to the partition $A'A''$ vs $C'C''$ because $\sigma_{A}^\lambda,\sigma_{C}^\lambda$ are real quantum states and $U, V$ are real quantum operations. The penultimate state, however, is not real separable (but is complex separable) because it violates a necessary condition for real separability~\cite{real_entanglement}: indeed, the $A''C''$ marginal state is not invariant under the partial transposition of either subsystem. This completes the proof.

\section{Proof of Theorem \ref{approximate_case}} \label{sec:noisy}

First, we need to solve the optimization problem
\be
\min_{\tau \in \mathcal{S}}\norm{\tau-\rho_0}_1,
\label{optim_PT}
\ee
\noindent where $\mathcal{S}$ is the set of states in $A'A''C'C''$ invariant under transposition of $A'A''$ and $\rho_0=\rho$, as defined by the right-hand side of eq. (\ref{suma}).

Let $S_{A'A''}$ be the quantum channel defined by $S_{A'A''}(\bullet)=\sum_{\alpha',\alpha''=\pm i}\proj{\alpha'}\otimes\proj{\alpha''} \bullet \proj{\alpha'}\otimes\proj{\alpha''}$. Note that $S_{A'A''}^2=S_{A'A''}$, that $S_{A'A''}\otimes\id_{C'C''}(\rho_0)=\rho_0$ and that $S_{A'A''}\otimes\id_{C'C''}(\mathcal{S})\subset \mathcal{S}$. By the monotonicity of the trace norm, we thus have that $\tau$ in \eqref{optim_PT} can be chosen invariant under $S_{A'A''}\otimes\id_{C'C''}$. 

Any operator $O$ invariant under $S_{A'A''}\otimes\id_{C'C''}$ is of the form 

\be
O=\sum_{\alpha',\alpha''=\pm i}\proj{\alpha'}_{A'}\otimes\proj{\alpha''}_{A''}\otimes O^{\alpha',\alpha''}_{C'C''}.
\ee
\noindent It can be verified that the partial transposition of systems $A',A''$ of any such operator effects the transformation $O\to(\sigma^{\otimes 2}_z\otimes\id^{\otimes 2}) O(\sigma^{\otimes 2}_z\otimes\id^{\otimes 2})$. In particular, it does not change the operator's spectrum. It follows that, for any state $\tau\in \mathcal{S}$, invariant under $S_{A'A''}\otimes \id_{C'C''}$,

\be
\|\rho_0-\tau\|_1=\|\rho^{T_{AA'}}_0-\tau^{T_{AA'}}\|_1=\|\rho^{T_{A'A''}}_0-\tau\|_1.
\ee

\noindent Invoking the triangle inequality, we have that

\be
2=\|\rho_0-\rho^{T_{A'A''}}\|_1\leq \|\rho_0-\tau\|_1+\|\rho^{T_{A'A''}}_0-\tau\|_1=2\|\rho_0-\tau\|_1,
\ee
\noindent and hence the solution of problem (\ref{optim_PT}) is lower bounded by $1$. This bound happens to be tight: it is saturated by taking $\tau$ to be the maximally mixed state.

Having solved problem (\ref{optim_PT}), we proceed to prove Theorem \ref{approximate_case}. We use the same notation as in section \ref{sec:pbar}. Define $\omega:=(\tr_{B_1B_2}\psi)\otimes\ketbra{0000}{0000}_{A'C'A''C''}$ as Alice and Charlie's state before any measurement by Bob. Consider the state $\rho_\varepsilon:=\tr_{{AC}}(U\otimes V \omega U^T\otimes V^T)$ left on the $A'A''C'C''$ systems after applying over Alice and Charlie's systems ${AC}$ the real local isometries $U$, $V$ (see Figure~\ref{isometry_pic}) and tracing out ${AC}$.  As reasoned, such a state must be invariant under transposition of the systems $A'A''$ if Alice, Bob and Charlie's system admits a real quantum representation. Obviously, that will not happen if

\be
\norm{\rho_{\varepsilon}-\rho_0}_1 < \min_{\tau \in \mathcal{S}}\norm{\tau-\rho_0}_1=1.
\label{trace_norm_arg}
\ee

By the triangle inequality, we can bound
\begin{equation}\label{eq:normbnd}
\norm{\rho_{\varepsilon}-\rho_0}_1\leq \epsilon_1(\varepsilon)+\epsilon_2(\varepsilon),
\end{equation}
in terms of $\epsilon_1(\varepsilon) :=\norm{\rho_{\varepsilon}-\sum_b P(b) \sigma^b}_1$ and $\epsilon_2(\varepsilon) :=\norm{\sum_b P(b) \sigma^b-\rho_0}_1$, where $\sigma^b= \tr_{{AC}}(\ketbra{\sigma^b}{\sigma^b})$ and $\ket{\sigma^b}$ are the potentially {\em unnormalized} states

\begin{align}\label{sigma}
    \ket{\sigma^b} = 
    \begin{cases}
    \ket{\phi^+}_{A'C'}\otimes\frac{1}{\sqrt{2}}\left[\ket{\psi^+}_{A''C''}\frac{\id+Z^A}{\sqrt{2}}\ket{\psi^b}_{AC} + \ket{\phi^-}_{A''C''}Y^AX^A\frac{\id+Z^A}{\sqrt{2}}\ket{\psi^b}_{AC}\right] & \text{ for } b = 00 = \phi^+\\
    \ket{\psi^+}_{A'C'}\otimes\frac{1}{\sqrt{2}}\left[\ket{\psi^+}_{A''C''}\frac{X^A(\id-Z^A)}{\sqrt{2}}\ket{\psi^b}_{AC} + \ket{\phi^-}_{A''C''}\frac{Y^A(\id-Z^A)}{\sqrt{2}}\ket{\psi^b}_{AC}\right] & \text{ for } b = 01 = \psi^+\\
    \ket{\phi^-}_{A'C'}\otimes\frac{1}{\sqrt{2}}\left[\ket{\psi^+}_{A''C''}\frac{\id+Z^A}{\sqrt{2}}\ket{\psi^b}_{AC} + \ket{\phi^-}_{A''C''}Y^AX^A\frac{\id+Z^A}{\sqrt{2}}\ket{\psi^b}_{AC}\right] & \text{ for } b = 10 = \phi^-\\
    \ket{\psi^-}_{A'C'}\otimes\frac{1}{\sqrt{2}}\left[\ket{\psi^+}_{A''C''}\frac{X^A(\id-Z^A)}{\sqrt{2}}\ket{\psi^b}_{AC} + \ket{\phi^-}_{A''C''}\frac{Y^A(\id-Z^A)}{\sqrt{2}}\ket{\psi^b}_{AC}\right] & \text{ for } b = 11 = \psi^-.
    \end{cases}
\end{align}
Thus, if the sum of the upper bounds on $\epsilon_1(\varepsilon)$ and $\epsilon_2(\varepsilon)$ computed in the following sections is smaller than $1$, then, by eq. (\ref{trace_norm_arg}), $\rho_\varepsilon$ is not reproducible with real quantum states. This leads to a critical error value of $\varepsilon_c = 7.18 \cdot 10^{-5}$.

\subsection{Upper bounds on $\epsilon_1(\varepsilon)$}

%Let $b=(b_1,b_2)$ with $b_1, b_2 \in \{0,1\}$ define the four inequalities 
%\begin{equation}
%\begin{aligned}
%    \mathscr{T}(b)=(-1)^{b_2}Z^A(D_{z,x}^C+E_{z,x}^C)&+(-1)^{b_1}X^A(D_{z,x}^C-E_{z,x}^C)+\\
%    (-1)^{b_2}Z^A(D_{z,y}^C+E_{z,y}^C)& - (-1)^{b_1+b_2}Y^A(D_{z,y}^C-E_{z,y}^C)+\\
%    (-1)^{b_1}X^A(D_{x,y}^C+E_{x,y}^C)& - (-1)^{b_1+b_2}Y^A(D_{x,y}^C-E_{x,y}^C)\leq6\sqrt{2}.
%\end{aligned}
%\end{equation}

\begin{lem} \label{lem1} Let  $\ket{\psi}$ be a state that, with the measurement operators $X^A$, $Y^A$, $Z^A$ for Alice and $D^C_{xy}$, $D^C_{zx}$, $D^C_{zy}$, $E^C_{xy}$, $E^C_{zx}$, $E^C_{zy}$ for Charlie, obeys $\bra{\psi}\hat{\mathscr{T}}_b\ket{\psi}=6\sqrt{2}-\varepsilon$, then
\begin{align}
    \norm{U\otimes V\left(\ket{\psi}_{{AC}}\ket{0000}_{A'C'A''C''}\right) - \ket{\sigma^b}}  \leq  (15+13\sqrt{2})\varepsilon_1,
    \label{eps2}
\end{align}
with $\varepsilon_1=\sqrt{\sqrt{2}\varepsilon}$ and where $U \otimes V$ denotes the isometry from Figure~\ref{isometry_pic} with $\hat{X}^C$, $\hat{Y}^C$, $\hat{Z}^C$ the regularised versions of $\frac{D^C_{zx}-E^C_{zx}}{\sqrt{2}}$, $\frac{D^C_{zy}-E^C_{zy}}{\sqrt{2}}$ and  $\frac{D^C_{zx}+E^C_{zx}}{\sqrt{2}}$ respectively.
\end{lem}

\begin{proof}
Let
\begin{align}
    \hat Z^C := \text{reg}(Z^C) \equiv \text{reg}(Z^C_{zx}), \text{ similarly for } X^C \equiv X^C_{zx} \text{ and } Y^C\equiv Y^C_{zy},
\end{align}
and
\begin{align}
    Z^C_{zx}:=\frac{D^C_{zx}+E^C_{zx}}{\sqrt{2}},\quad X^C_{zx}:=\frac{D^C_{zx}-E^C_{zx}}{\sqrt{2}}\\
    Z^C_{zy}:=\frac{D^C_{zy}+E^C_{zy}}{\sqrt{2}},\quad Y^C_{zy}:=\frac{D^C_{zy}-E^C_{zy}}{\sqrt{2}}\\
    X^C_{xy}:=\frac{D^C_{xy}+E^C_{xy}}{\sqrt{2}},\quad Y^C_{xy}:=\frac{D^C_{xy}-E^C_{xy}}{\sqrt{2}}
\end{align}
Note that these definitions are independent of $b$. Then using the sum-of-square (SOS) decomposition
\begin{align}
    \sqrt{2}(6\sqrt{2}-\hat{\mathscr{T}}_b)&=\left[(-1)^{b_2}Z^A-\frac{D^C_{zx}+E^C_{zx}}{\sqrt{2}}\right]^2 + \left[(-1)^{b_1}X^A-\frac{D^C_{zx}-E^C_{zx}}{\sqrt{2}}\right]^2\\
    &+\left[(-1)^{b_2}Z^A-\frac{D^C_{zy}+E^C_{zy}}{\sqrt{2}}\right]^2 + \left[(-1)^{b_1+b_2}Y^A+\frac{D^C_{zy}-E^C_{zy}}{\sqrt{2}}\right]^2\\
    &+\left[(-1)^{b_1}X^A-\frac{D^C_{xy}+E^C_{xy}}{\sqrt{2}}\right]^2 + \left[(-1)^{b_1+b_2}Y^A+\frac{D^C_{xy}-E^C_{xy}}{\sqrt{2}}\right]^2.
\end{align}
we read off the following approximate relations
\begin{equation}\label{eq:CtoA}
\norm{((-1)^{b_2}Z^A-Z^C)\ket{\psi}}, \ \norm{((-1)^{b_1}X^A-X^C)\ket{\psi}}, \ \norm{((-1)^{b_1+b_2}Y^A+Y^C)\ket{\psi}} \leq \varepsilon_1\,.
\end{equation}

Using the SOS decomposition
\begin{align}
\sqrt{2}(6\sqrt{2}\id -\hat{\mathscr{T}}_b) &= \left[D^C_{zx}- \frac{(-1)^{b_2}Z^A+(-1)^{b_1}X^A}{\sqrt{2}}\right]^2+\left[E^C_{zx} - \frac{(-1)^{b_2}Z^A-(-1)^{b_1}X^A}{\sqrt{2}} \right]^2 \nonumber \\
&+ \left[D^C_{zy} - \frac{(-1)^{b_2}Z^A-(-1)^{b_1+b_2}Y^A}{\sqrt{2}}\right]^2+\left[E^C_{zy} - \frac{(-1)^{b_2}Z^A+(-1)^{b_1+b_2}Y^A}{\sqrt{2}} \right]^2 \nonumber \\
&+ \left[D^C_{xy} - \frac{(-1)^{b_1}X^A-(-1)^{b_1+b_2}Y^A}{\sqrt{2}}\right]^2+\left[E^C_{xy} - \frac{(-1)^{b_1}X^A+(-1)^{b_1+b_2}Y^A}{\sqrt{2}} \right]^2 \label{sos2}
\end{align}
we can prove that
\begin{equation}\label{eq:anticomm_approx}
    \norm{\{(-1)^{b_2}Z^A,(-1)^{b_1}X^A\}\ket{\psi}}, \ \norm{\{(-1)^{b_2}Z^A,(-1)^{b_1+b_2}Y^A\}\ket{\psi}}, \ \norm{\{(-1)^{b_1+b_2}Y^A,(-1)^{b_1}X^A\}\ket{\psi}} \leq 2(1+\sqrt{2})\varepsilon_1,
\end{equation}
i.e., $X^A,Y^A,Z^A$ all anticommute approximately, as follows:
Since $(D^C_{zx})^2=\id$ we get
\begin{align}
    \left[D^C_{zx}+ \frac{(-1)^{b_2}Z^A+(-1)^{b_1}X^A}{\sqrt{2}}\right]\left[D^C_{zx}- \frac{(-1)^{b_2}Z^A+(-1)^{b_1}X^A}{\sqrt{2}}\right] =  -\frac{\{(-1)^{b_2}Z^A,(-1)^{b_1}X^A\}}{2}.
\end{align}
Now apply both sides to $\ket{\psi}$ and take the norm; we get the desired inequality after noticing that the operator norm of the first square bracket is bounded by $1+\sqrt{2}$ (by triangle inequality and $D^C_{zx}$ being unitary).
Finally, the regularized operators are also close to the unregularized counterparts, e.g.
\begin{equation}\label{eq:reg_approx}
    \norm{(\hat Z^C-Z^C)\ket{\psi}}=\norm{(\id-(\hat Z^C)^\dagger Z^C)\ket{\psi}}=\norm{(\id-\abs{Z^C})\ket{\psi}}=\norm{(\id-\abs{Z^AZ^C})\ket{\psi}}\leq\norm{(\id-Z^AZ^C)\ket{\psi}}\leq\varepsilon_1.
\end{equation}

Now we apply the isometry $U\otimes V$ defined in Figure~\ref{isometry_pic}. The cancellation happens exactly as in the ideal case  incurring a small loss measured in vector norm because relations are only approximate. After $U'\otimes V'$ (see Figure~\ref{isometry_pic}), the unknown state is close to a (potentially unnormalized) vector in a specific form
\begin{equation}
    \norm{U'\otimes V'(\ket{\psi}_{{AC}}\ket{00}_{A'C'})-\mathcal{O}(b)\ket{\psi}_{{AC}}\ket{b}_{A'C'}} \leq (5+\sqrt{2})\varepsilon_1
\end{equation}
where
\begin{equation}
    \mathcal{O}(b)=\begin{cases}
    \frac{\id+Z^A}{\sqrt{2}} & \text { if } b_2=0\\
    \frac{X^A(\id-Z^A)}{\sqrt{2}} & \text { if } b_2=1
    \end{cases}
\end{equation}
To get this result, we compute the effect of the different steps of the isometry on the initial vector and make use of the relations \eqref{eq:CtoA}, \eqref{eq:anticomm_approx}, \eqref{eq:reg_approx} in order to simplify the final expression, all the while keeping track of the error incurred to at each stage. That is, we first apply the Hadamard gates, and then control $Z^A$ and $\hat Z^C$ gates, which return a state as written in the first line of Step 1 below. Since so far we did not use any of the approximate relations, the error incurred to in this step is $0$, as written on the left of the state in Step 1. Next, we use the approximate relation \eqref{eq:CtoA} to convert $Z^A\hat Z^C\ket{\psi}$ to $(-1)^{b_2}Z^AZ^A\ket{\psi}=(-1)^{b_2}\ket{\psi}$, and similarly $\hat Z^C\ket{\psi}$ to $\hat Z^A\ket{\psi}$, incurring into an error $2\varepsilon_1$, written on the left. Continuing this way through the circuit, the intermediate expressions, together with their bounds, are the following:
\begin{align*}
  \text{Step 1:} \qquad \qquad 0\varepsilon_1: & \frac{1}{2}\left[\ket{00}+\ket{11}Z^A\hat Z^C+\ket{01}\hat Z^C+\ket{10}Z^A\right]\ket{\psi}\\
    2\varepsilon_1: & \frac{1}{2}\left[(\ket{00}+(-1)^{b_2}\ket{11})\id+((-1)^{b_2}\ket{01}+\ket{10})Z^A\right]\ket{\psi}\\
  \text{Step 2:} \qquad \qquad 0\varepsilon_1: & \frac{1}{4}\Big[\ket{00}(1+(-1)^{b_2})(\id+Z^A)+\ket{11}(1+(-1)^{b_2})X^A\hat X^C(\id-Z^A)+\\
    \, & \quad \quad \ket{01}(1-(-1)^{b_2})\hat X^C(\id+Z^A)+\ket{10}(1-(-1)^{b_2})X^A(\id-Z^A)]\Big]\ket{\psi}\\
    (3+\sqrt{2})\varepsilon_1: & \frac{1}{4}\left[(1+(-1)^{b_2})(\ket{00}+(-1)^{b_1}\ket{11})(\id+Z^A) + (1-(-1)^{b_2})((-1)^{b_1}\ket{01}+\ket{10})X^A(\id-Z^A)]\right]\ket{\psi}
\end{align*}
Note that in the last approximation, the two paths $b_2=0$ and $b_2=1$ have the same upper bound. Also, we use the convention that $\ket{\psi^-}=(\ket{10}-\ket{01})/\sqrt{2}$, which has the same density matrix as the usual convention.

The rest of the circuit does not involve $A'C'$ so we can safely ignore them. After the remaining local unitaries $U''\otimes V''$ (see Figure~\ref{isometry_pic}), 
\begin{equation}
\norm{U''\otimes V''\left(\mathcal{O}(b)\ket{\psi}_{{AC}}\ket{00}_{A''C''}\right) - \ket{\tau^b}} \leq (10+12\sqrt{2})\varepsilon_1,
\end{equation}
where
\begin{align}
    \ket{\tau^b}=
    \begin{cases}
    \frac{1}{\sqrt{2}}\left[\ket{\psi^+}\frac{\id+Z^A}{\sqrt{2}}\ket{\psi} + \ket{\phi^-}Y^AX^A\frac{\id+Z^A}{\sqrt{2}}\ket{\psi}\right] & \text{ for } b_2 = 0\\
    \frac{1}{\sqrt{2}}\left[\ket{\psi^+}\frac{X^A(\id-Z^A)}{\sqrt{2}}\ket{\psi} + \ket{\phi^-}\frac{Y^A(\id-Z^A)}{\sqrt{2}}\ket{\psi}\right] & \text{ for } b_2 = 1
    \end{cases}
\end{align}
The intermediate steps are
\begin{align*}
  \text{Step 3:} \qquad \qquad  \qquad 0\varepsilon_1 : & \frac{1}{2}\left[\ket{00}\mathcal{O}(b) + \ket{11}Y^AX^A\hat Y^C\hat X^C\mathcal{O}(b) + \ket{01}\hat Y^C\hat X^C\mathcal{O}(b) + \ket{10}Y^AX^A\mathcal{O}(b)\right]\ket{\psi}\\
    (8+10\sqrt{2})\varepsilon_1: & \frac{1}{2}\left[\ket{00}\mathcal{O}(b)  + \ket{11}Y^AX^AY^AX^A\mathcal{O}(b) +  \ket{01}Y^A X^A\mathcal{O}(b) + \ket{10}Y^AX^A\mathcal{O}(b)\right]\ket{\psi}\\
    (2+2\sqrt{2})\varepsilon_1: & \frac{1}{2}\left[(\ket{00} - \ket{11})\mathcal{O}(b) + (\ket{01} + \ket{10})Y^AX^A\mathcal{O}(b)\right]\ket{\psi}\\
   \text{Step 4:} \qquad \qquad  \qquad 0\varepsilon_1: & \frac{1}{2}\left[(\ket{01}+\ket{10})\mathcal{O}(b) + (\ket{00}-\ket{11})Y^AX^A\mathcal{O}(b)\right]\ket{\psi}
\end{align*}
where we have taken the larger bound among the two $b_2=1$ and $b_2=0$ cases to get
\begin{align}
    \norm{(\hat Y^C\hat X^C\mathcal{O}(b) - Y^AX^A\mathcal{O}(b))\ket{\psi}} &\leq (20+8\sqrt{2})\varepsilon_1/\sqrt{2} \text{ and }\\% b2=0 larger
    \norm{(Y^AX^AY^AX^A\mathcal{O}(b) + \mathcal{O}(b))\ket{\psi}} &\leq (8+4\sqrt{2})\varepsilon_1/\sqrt{2}\,.% same for both b2
\end{align}
By a series of triangle inequalities going through all the intermediate expressions, we get the Lemma.
%\begin{align}
%    \norm{U\otimes V\left(\ket{\psi}^{AC}\ket{0000}^{A'C'A''C''}\right) - \ket{\sigma^b}}  \leq  (10.5+11\sqrt{2}+(\sqrt{2})^{-1})\varepsilon_1.
%    \label{eps2}
%\end{align}
%where for $\ket{\psi}\equiv\ket{\psi^b}$ the subnormalized state per our convention,
%\begin{align}\label{eq: state}
%    \ket{\sigma^b} = \ket{b}^{A'C'}\otimes\ket{\tau^b}^{ACA''C''} = 
%    \begin{cases}
%    \ket{\Phi^+}\otimes\frac{1}{\sqrt{2}}\left[\ket{\Psi^+}\frac{\id+Z^A}{\sqrt{2}}\ket{\psi} + \ket{\Phi^-}Y^AX^A\frac{\id+Z^A}{\sqrt{2}}\ket{\psi}\right] & \text{ for } b = 00 = \Phi^+\\
%    \ket{\Psi^+}\otimes\frac{1}{\sqrt{2}}\left[\ket{\Psi^+}\frac{X^A(\id-Z^A)}{\sqrt{2}}\ket{\psi} + \ket{\Phi^-}\frac{Y^A(\id-Z^A)}{\sqrt{2}}\ket{\psi}\right] & \text{ for } b = 01 = \Psi^+\\
%    \ket{\Phi^-}\otimes\frac{1}{\sqrt{2}}\left[\ket{\Psi^+}\frac{\id+Z^A}{\sqrt{2}}\ket{\psi} + \ket{\Phi^-}Y^AX^A\frac{\id+Z^A}{\sqrt{2}}\ket{\psi}\right] & \text{ for } b = 10 = \Phi^-\\
%    \ket{\Psi^-}\otimes\frac{1}{\sqrt{2}}\left[\ket{\Psi^+}\frac{X^A(\id-Z^A)}{\sqrt{2}}\ket{\psi} + \ket{\Phi^-}\frac{Y^A(\id-Z^A)}{\sqrt{2}}\ket{\psi}\right] & \text{ for } b = 11 = \Psi^-
 %   \end{cases}
%\end{align}
\end{proof}

\begin{lem} \label{lem2} Let $\ket{\sigma^b}$ be any state defined in \eqref{sigma} and let $\varepsilon_1=\sqrt{\sqrt{2}\varepsilon}$, then
\begin{equation}\label{eq:lem2}
1 - (3+\sqrt{2})\varepsilon_1 \leq \norm{\ket{\sigma^b}}^2 \leq 1 + (3+\sqrt{2})\varepsilon_1
\end{equation}
\end{lem}

\begin{proof}
With $\norm{\ket{\sigma^b}}^2 = \tr{\proj{\sigma^b}}$ then
\be\label{eq:lem2b}
\abs{1-\tr{\proj{\sigma^b}}} = \abs{\bra{\psi^b} Z^A \ket{\psi^b}}.
\ee

\noindent  Now consider, as in \cite{singlet},
\begin{align*}
\abs{\bra{\psi^b}Z^A \ket{\psi^b}+\bra{\psi^b}Z^A X^A \hat{X}^C \ket{\psi^b}}
&=\abs{\bra{\psi^b}Z^A \ket{\psi^b}+\bra{\psi^b}\hat{X}^C X^A Z^A \ket{\psi^b}}\\
 &= \abs{\bra{\psi^b}(\hat{X}^C Z^A \hat{X}^C + \hat{X}^C X^A Z^A) \ket{\psi^b}} \\
&\leq \norm{\hat{X}^C \ket{\psi^b}} \norm{(Z^A \hat{X}^C +X^A Z^A) \ket{\psi^b}} \leq  (4+2\sqrt{2})\varepsilon_1. \\
\abs{\bra{\psi^b}Z^A \ket{\psi^b}-\bra{\psi^b}Z^AX^A\hat{X}^C \ket{\psi^b}} &\leq \norm{(Z^A(\id - X^A \hat{X}^C)\ket{\psi^b}}\\
& \leq \norm{Z^A(\hat{X}^C-X^A) \ket{\psi^b}} \leq 2\varepsilon_1,
\end{align*}
where we used the Cauchy-Schwarz inequality as well as 
\[
\norm{(Z^A\hat{X}^C+ X^AZ^A) \ket{\psi^b}} \leq \norm{(Z^A\hat{X}^C-Z^AX^A)\ket{\psi^b}} + \norm{\{X^A,Z^A\}\ket{\psi^b}} \leq 2\varepsilon_1+2(1+\sqrt{2})\varepsilon_1.
\]
Therefore, we find that \begin{equation}\label{ZA}
 \abs{\bra{\psi^b} Z^A \ket{\psi^b}}\leq (3+\sqrt{2})\varepsilon_1,   
\end{equation} which combined with \eqref{eq:lem2b} implies \eqref{eq:lem2}.
\end{proof}

\bigskip

\noindent Using these Lemmas, we can obtain the final bound for $\epsilon_1(\varepsilon)$. Firstly, we have
\begin{equation}
\epsilon_1(\varepsilon) \leq \sum_b P(b) \norm{\ketbra{\rho_\varepsilon^b}{\rho_\varepsilon^b}-\ketbra{\sigma^b}{\sigma^b}}_1  \leq  2 \sum_b P(b) \sqrt{
\frac{(1+\norm{\ket{\sigma^b}}^2)^2}{4}-\abs{\braket{\rho_\varepsilon^b}{\sigma^b}}^2},
\end{equation}
where the first equality comes from $\rho_\varepsilon=\sum_b P(b) \rho_\varepsilon^b$ for $\rho_\varepsilon^b$ the $A'C'A''C''$ marginal of the state $\ket{\rho_\varepsilon^b}=U\otimes V(\ket{\psi^b}_{AC}\ket{0000}^{A'C'A''C''})$.
The last inequality follows by adapting an argument from~\cite{Watrous} to non-normalised states. Namely, since $\ketbra{\rho_\varepsilon^b}{\rho_\varepsilon^b}-\ketbra{\sigma^b}{\sigma^b}$ has rank at most two,
\begin{alignat}{2}
\lambda_1+\lambda_2&= \  \tr(\ketbra{\rho_\varepsilon^b}{\rho_\varepsilon^b}-\ketbra{\sigma^b}{\sigma^b}) &&= 1- \norm{\ket{\sigma^b}}^2 \nonumber\\
\lambda_1^2+\lambda_2^2 &= \ \tr((\ketbra{\rho_\varepsilon^b}{\rho_\varepsilon^b}-\ketbra{\sigma^b}{\sigma^b})^2) &&= 1+ \norm{\ket{\sigma^b}}^4-2\abs{\braket{\rho_\varepsilon^b}{\sigma^b}}^2, \nonumber
\end{alignat}
where $\lambda_1$ and $\lambda_2$ are the two possibly non-zero eigenvalues of $\ketbra{\rho_\varepsilon^b}{\rho_\varepsilon^b}-\ketbra{\sigma^b}{\sigma^b}$. Solving explicitly for $\lambda_1$ and $\lambda_2$ in the previous system of equations shows that
$$\abs{\lambda_1}+\abs{\lambda_2}= 2\sqrt{
\frac{(1+\norm{\ket{\sigma^b}}^2)^2}{4}-\abs{\braket{\rho_\varepsilon^b}{\sigma^b}}^2}. $$

\noindent Furthermore, 
\[
\braket{\rho_\varepsilon^b}{\sigma^b} = \frac{1+\norm{\ket{\sigma^b}}^2-\norm{\ket{\rho_\varepsilon^b}-\ket{\sigma^b}}^2}{2} 
\]
which implies, by Lemma~\ref{lem1} and Lemma~\ref{lem2}, that for small $\varepsilon > 0$,
\begin{equation}
\epsilon_1(\varepsilon) \leq 2\sqrt{\frac{[1+(1+(3+\sqrt{2})\varepsilon_1)]^2}{4}-\left(1- \frac{(3+\sqrt{2})\varepsilon_1 + \varepsilon_2^2}{2} \right)^2} \text{ with } \varepsilon_2:=(15+13\sqrt{2})\varepsilon_1.
\end{equation}
%where we set $\varepsilon_2:={\color{red}(15+10\sqrt{2}+3/\sqrt{2})\varepsilon_1}$.

\subsection{Upper bounds on $\epsilon_2(\varepsilon)$ }

To bound $\epsilon_2(\varepsilon)=\|\sum_b P(b)\sigma^b-\rho_0\|_1$, let us first separate the expression in the norm as 
\begin{align*}
\norm{\sigma^{b}-\rho_0^{b}}_1 &\leq \abs{\hat{\mu}(b)} \norm{\proj{b}\otimes(\ketbra{\phi^-}{\psi^+}+\ketbra{\psi^+}{\phi^-})}_1 + \abs{\mu(b)} \norm{\proj{b}\otimes(\proj{\phi^-}+\proj{\psi^+})}_1, \\
%&= \frac{1}{2}(3+\sqrt{2})\varepsilon_1+\frac{1}{2}(4\varepsilon_1 + (1+\sqrt{2})\varepsilon_1) = (4+\sqrt{2})\varepsilon_1
\end{align*}
where the coefficients simplify to $
\abs{\mu(b)} = \frac{1}{2} \abs{\bra{\psi^b} Z^A \ket{\psi^b}}$
and 
\begin{align*}
 \abs{\hat{\mu}(b)} =\begin{cases}
  \frac{1}{4} \abs{\bra{\psi^b} (\id + Z^A) X^A Y^A(\id + Z^A) \ket{\psi^b}} & b_2=0 \\ 
  \frac{1}{4} \abs{\bra{\psi^b} (\id - Z^A) X^A Y^A(\id - Z^A) \ket{\psi^b}} & b_2=1. \\
 \end{cases}  
\end{align*} 

\noindent To bound $\abs{\hat{\mu}(b)}$, let us consider
\begin{align}
 \frac{1}{4} \abs{\bra{\psi^b} (\id + (-1)^{b_2} Z^A)\{X^A, Y^A\}(\id + (-1)^{b_2} Z^A) \ket{\psi^b}} & \leq \frac{1}{4}\norm{(\id+ (-1)^{b_2} Z^A)\ket{\psi^b}}\norm{\{X^A, Y^A\}(\id + (-1)^{b_2}  Z^A) \ket{\psi^b}}  \nonumber \\
&\leq \frac{1}{2} \norm{\{X^A, Y^A\}(\id + (-1)^{b_2} Z^A) \ket{\psi^b}} \nonumber \\
& \leq \frac{1}{2} (\norm{\{X^A, Y^A\}\ket{\psi^b}} \nonumber\\
&+ \norm{\hat{Z}^C\{X^A, Y^A\} \ket{\psi^b}} + \norm{\{X^A, Y^A\}((-1)^{b_2}Z^A-\hat{Z}^C)\ket{\psi^b}}) \nonumber \\
& \leq  (4+2\sqrt{2})\varepsilon_1.\label{eq:b1}
\end{align}

\noindent Therefore, using \eqref{eq:b1} and \eqref{ZA}, we obtain
\begin{align*}
\norm{\sigma^{b}-\rho_0^{b}}_1 &\leq 2(\abs{\hat\mu(b)}+\abs{\mu(b)})\leq
 (4+2\sqrt{2})\varepsilon_1+(3+\sqrt{2})\varepsilon_1 = (7+3\sqrt{2})\varepsilon_1.
\end{align*}

\noindent Since $\abs{P(b)-\frac{1}{4}} \leq \varepsilon$, we obtain 
\begin{align}
\epsilon_2(\varepsilon) \leq \sum_b \frac{1}{4} \norm{\sigma^b-\rho_0^b}_1 + \sum_b \abs{P(b)-\frac{1}{4}} \norm{ \sigma^b }_1\leq (7+3\sqrt{2})\varepsilon_1+4(1 + (3+\sqrt{2})\varepsilon_1)\epsilon 
\end{align}

\section{Proof of Theorem \ref{main_theo}}
\label{numerics}

Our first step is to obtain a tractable necessary condition for a distribution $P$ to admit a real quantum realization. To this aim, we invoke the essentials of non-commutative polynomial optimization \cite{NOP}.

Let $P$ admit a real quantum representation of the form (\ref{decomp}), and think of an abstract unital $\star$-algebra ${\cal A}$ with generators $A_{1|1},A_{1|2},A_{1|3}$, together with the relations $A_{1|x}=A_{1|x}^\dagger=(A_{1|x})^2$, for $x=1,2,3$. We consider the natural isomorphism $\pi_A$ that maps any element of ${\cal A}$ to the algebra generated by the operators $\{\tilde{A}_{1|x}:x=1,2,3\}$. That is, $A_{1|x}$ is an element of ${\cal A}$, while $\pi_A(A_{1|x})=\tilde{A}_{1|x}$ represents Alice's physical measurement operator for setting $x$ and outcome $a=1$. We call $\A\subset {\cal A}$ the set of all monomials of degree $n_A$ or lower of $A_{1|1},A_{1|2},A_{1|3}\in {\cal A}$ including the identity, for some fixed natural number $n_A$. Similarly, we define an abstract unital $\star$-algebra ${\cal C}$ with projection generators $C_{1|1},...,C_{1|6}$ to model Charlie's measurements, with the natural isomorphism $\pi_C$. We denote by $\C\subset{\cal C}$ the set of all monomials of degree $n_C$ or lower of $C_{1|1},...,C_{1|6}$, including the identity for some fixed $n_C$.

Following Moroder \emph{et al.}'s interpretation \cite{moroder} of the Navascu\'es-Pironio-Ac\'in (NPA) hierarchy \cite{npa, npa2}, to each monomial $\alpha\in {\cal A}$ we associate a normalized ket $\ket{\alpha}$ such that $\braket{\alpha'}{\alpha}=\delta_{\alpha,\alpha'}$, for all $\alpha,\alpha'\in{\cal A}$. Likewise, we associate an orthonormal basis to the set ${\cal C}$. Next, we define the completely positive local maps $\Omega_A$, $\Omega_C$ through the relations:

\begin{align}
&\Omega_A(\eta)=\sum_{\alpha,\alpha'\in\A} \tr\left((\pi_A(\alpha)^{\dagger}\eta \pi_A(\alpha')\right)\ket{\alpha}\bra{\alpha'}\nonumber\\
&\Omega_C(\eta)=\sum_{\gamma,\gamma'\in \C} \tr\left(\pi_C(\gamma)^\dagger\eta \pi_C(\gamma')\right)\ket{\gamma}\bra{\gamma'}.
\end{align}

\noindent $\Omega_A$ ($\Omega_C$) thus maps states in Alice's (Charlie's) untrusted system $A$ ($C$) to a non-normalized state with support in $H_A=\mbox{span}\{\ket{\alpha}:\alpha\in\A\}$ ($H_C=\mbox{span}\{\ket{\gamma}:\gamma\in \C\}$).

For $\lambda\in\Lambda$, define the state $\tilde{\omega}^b(\lambda)\equiv\tr_{B_1B_2}\left\{(\tilde{\sigma}_{AB_1}^\lambda\otimes \tilde{\sigma}_{B_2C}^\lambda)(\id_A\otimes \tilde{B}_b\otimes\id_C)\right\}$ and consider the $|\A||\C|\times |\A||\C|$ matrix

\be
\Gamma^b\equiv \sum_\lambda P(\lambda)(\Omega_A\otimes\Omega_C)(\tilde{\omega}^b(\lambda)).
\ee
\noindent Since $\Omega_A,\Omega_C$ are completely positive, this matrix must be positive semidefinite. Moreover, some of its entries are related. Indeed, let the monomials $\alpha_1,...,\alpha_4\in\A$, $\gamma_1,...,\gamma_4\in\C$ be such that $\alpha_2\alpha_1^\dagger=\alpha_4\alpha_3^\dagger=:\alpha$, $\gamma_2\gamma_1^\dagger=\gamma_4\gamma_3^\dagger=:\gamma$. Then it holds that

\be
\bra{\alpha_1}\bra{\gamma_1}\Gamma^b\ket{\alpha_2}\ket{\gamma_2}=\bra{\alpha_3}\bra{\gamma_3}\Gamma^b\ket{\alpha_4}\ket{\gamma_4}=\sum_{\lambda} P(\lambda)\tr\{\tilde{\omega}^b(\lambda) (\pi_A(\alpha)\otimes\pi_C(\gamma))\}.
\label{basic_identity}
\ee
\noindent This allows us to write $\Gamma^b$ as

\be
\Gamma^b=\sum_{\alpha\in\A\cdot\A,\gamma\in\C\cdot\C}d^b_{\alpha,\gamma}M^{\alpha}\otimes N^{\gamma},
\label{gammas}
\ee
\noindent where $\{d^b_{\alpha,\gamma}:\alpha,\gamma\}$ are real coefficients, and, for any $a,a'\in\A, c,c'\in\C$, the corresponding entries of $M^{\alpha}$, $N^{\gamma}$ are given by

\be
M^\alpha_{a,a'}=\delta_{\alpha, a'a^\dagger},N^\gamma_{c,c'}=\delta_{\gamma,c'c^\dagger}.
\ee
\noindent Also, notice that some of the coefficients $d^b_{\alpha,\gamma}$ follow from the experimental data $P(a,b,c|x,z)$. Namely, 

\be
d^b_{\id,\id}=P(b),d^b_{A_{1|x},\id}=P_{AB}(1,b|x),d^b_{\id,C_{1|z}}=P_{BC}(b,1|z),d^b_{A_{1|x},C_{1|z}}=P(1,b,1|x,z).
\label{probas}
\ee
\noindent In the terminology of \cite{NOP}, $\Gamma^b$ is a non-normalized \emph{moment matrix} for the distribution $P(a,c|x,z, b)$, with norm $d^b_{\id,\id}=P(b)$.

Consider now the matrix $\Gamma\equiv\sum_b\Gamma^b$. From the definition of $\tilde{\omega}^b(\lambda)$, we have that, for any $\lambda\in\Lambda$, $\sum_b\tilde{\omega}^b(\lambda)=\tilde{\sigma}^\lambda_A\otimes \tilde{\sigma}^\lambda_C$. It follows that

\be
\Gamma=\sum_{\lambda}P(\lambda)\Omega_A(\tilde{\sigma}^\lambda_A)\otimes\Omega_C(\tilde{\sigma}^\lambda_C).
\ee
Since $\Omega_A, \Omega_C$ are real completely positive maps, we have that $\Gamma$ is a real separable operator, i.e., a conic combination of real product quantum states. In particular, it must be that $\Gamma$ equals its partial transpose, $\Gamma^{T_A}=\Gamma$ \cite{real_entanglement}.

To summarize: if $P$ admits a real quantum representation in the SWAP scenario, then there must exist real coefficients $d^b_{\alpha,\gamma}$ such that eq. (\ref{probas}) holds, the matrices $\Gamma^b$ defined through eq. (\ref{gammas}) are positive semidefinite and the matrix $\sum_{b}\Gamma^b$ is its own partial transpose. Consider then the following optimization problem:

\begin{align}
\max_{d, P} \quad &\mathscr{T}(P),\nonumber\\
\mbox{such that } \quad &\Gamma^b=\sum_{\alpha\in\A\cdot \A,\gamma\in\C\cdot \C}d^b_{\alpha,\gamma}M^\alpha\otimes N^\gamma\geq 0,\mbox{ for }b=1,...,4,\nonumber\\
&(\sum_b\Gamma^b)^{T_A}=\sum_b\Gamma^b,\nonumber\\
&d^b_{\id,\id}=P(b),d^b_{A_{1|x},\id}=P_{AB}(1,b|x),d^b_{\id,C_{1|z}}=P_{BC}(b,1|z),d^b_{A_{1|x},C_{1|z}}=P(1,b,1|x,z),\nonumber\\
&P(a,b,c|x,z)\geq 0, \sum_{a}P(a,b,c|x,z)=P_{BC}(b,c|z),\sum_{c}P(a,b,c|x,z)=P_{AB}(a,b|x),\nonumber\\&\sum_aP_{AB}(a,b|x)=\sum_cP_{BC}(b,c|z)=P(b), \sum_b P(b)=1,
\end{align}
\noindent where the conditions on $P(a,b,c|x,z)$ in the last two lines enforce that $P$ corresponds to a non-signalling, normalized tripartite distribution.

The above is a semidefinite program (SDP) \cite{sdp}, and, as long as the matrices $M^\alpha,N^\gamma$ are not very large, one can find the solution in a normal desktop. Since the constraints enforce a relaxation of the requirement that $P$ admit a representation in the SWAP scenario, it follows that the solution of this problem is an upper bound on the maximum value of $\mathscr{T}(P)$ for $P$ achievable through real quantum systems in the considered causal structure.

Taking $n_A=n_C=2$, and using the SDP solver MOSEK \cite{mosek} within the optimization package YALMIP \cite{yalmip}, we arrive at Theorem \ref{main_theo}.

\section{Experimental considerations}
\label{sec:experimental}

In view of Theorem \ref{main_theo}, there exists a considerable gap between the predictions of real and complex quantum theory in the entanglement swapping scenario. Does that mean that an actual experimental refutation of real quantum theory is within reach? Let us briefly consider how an experimental realization of the quantum experiment depicted in Figure \ref{fig:entswap} (up) would go. 

First, a general observation: even if our proposed experimental setup mimics that of Figure \ref{fig:entswap}, in principle, there could exist prior quantum correlations between the two preparation devices, or a quantum state shared by the three parties as a result of a past quantum interaction (say, in the last round of experiments). In either predicament, our previous bounds on the input-output statistics for real quantum system would not be valid, and any claims of refutation of real quantum physics would be unfounded. These two possibilities, though, rely on the presence of hidden quantum memories within the experimental equipment. Hence, if we posit a time-scale beyond which the devices' quantum memories degrade into classical information, we can discard such contingencies just by spacing out the experimental rounds sufficiently (note that prior classical correlations between the parties can be absorbed into the definition of $\lambda$). 

There is another loophole, namely, the possible existence of hidden state sources that distribute general tripartite entangled states at each experimental round. This would similarly compromise the conclusions of the experiment, and cannot be ruled out by appealing to decoherence or space-like separation. In order to refute real quantum physics, we are thus compelled to accept some plausible, yet unverifiable, assumptions about the form of the quantum states distributed to the three parties. In the following, we therefore \emph{postulate} that Figure \ref{fig:entswap} (down) accurately captures the causal scenario encountered by the three parties at each experimental round. Crucially, we allow the states $\tilde{\sigma}_{AB_1}^{\lambda}, \tilde{\sigma}_{B_2C}^\lambda$ and the distribution $P(\lambda)$ to depend in arbitrary ways on the past history $h$ of measurement settings and outcomes observed by the three parties in the course of the experiment. That is, we allow the real quantum physicist to adapt its states and measurement operators at each round to make us believe that it holds complex quantum resources.

Under the above adversarial conditions, the inequality $\mathscr{T}\leq 7.6605$ nonetheless holds at every experimental round. Therefore, one can use the techniques in \cite{Elkouss} and \cite{Mateus} to devise an $n$-round experiment that, if successful, disproves the hypothesis of real quantum physics with high statistical confidence.

Of course, in order to get there, one first needs to realize a quantum entanglement swapping experiment with $\mathscr{T}>7.6605$. We next discuss the technical feasibility of this goal. We assume that $\bar{\sigma}_{AB},\bar{\sigma}_{BC}$ are distributed via photon sources, and that photon polarization measurements are almost perfect. We then face two experimental problems: making sure that the photons reach their destination more or less unperturbed and conducting the Bell measurement. With regards to the first problem, we model the interaction between the photons and the environment through white noise. That is, rather than being distributed two maximally entangled states, Alice and Bob and Bob and Charlie respectively receive an independent copy of the state $v\Phi^++\frac{1-v}{4}\openone$. Under ideal Bell measurements, the violation of the inequality $\mathscr{T}\leq 7.6605$ thus requires each photon source to have a visibility $v$ of at least $v=\sqrt{\frac{7.6605}{6\sqrt{2}}}\approx 0.95$, a value realistic with present technology. With regards to the second problem, deterministic Bell-state measurements with photons are indeed complicated \cite{Bell_meas_impossible}, but they can be carried out with arbitrary precision, provided that sufficiently many single-photon sources are available \cite{Bell_meas_possible}. All in all, the experimental requirements to violate (\ref{violation}) are demanding, but within reach.

\end{appendix}

\end{document}